\documentclass[sigconf]{acmart}

\AtBeginDocument{%
  \providecommand\BibTeX{{%
    \normalfont B\kern-0.5em{\scshape i\kern-0.25em b}\kern-0.8em\TeX}}}

\makeatletter
\def\@ACM@checkaffil{
    \if@ACM@instpresent\else
    \ClassWarningNoLine{\@classname}{No institution present for an affiliation}%
    \fi
    \if@ACM@citypresent\else
    \ClassWarningNoLine{\@classname}{No city present for an affiliation}%
    \fi
    \if@ACM@countrypresent\else
        \ClassWarningNoLine{\@classname}{No country present for an affiliation}%
    \fi
}

\copyrightyear{2023} 
\acmYear{2023} 
\setcopyright{acmlicensed}\acmConference[WWW '23]{Proceedings of the ACM Web Conference 2023}{May 1--5, 2023}{Austin, TX, USA}
\acmBooktitle{Proceedings of the ACM Web Conference 2023 (WWW '23), May 1--5, 2023, Austin, TX, USA}
\acmPrice{15.00}
\acmDOI{10.1145/3543507.3583229}
\acmISBN{978-1-4503-9416-1/23/04}

\usepackage[super]{nth}
\usepackage{amsmath}
\usepackage{amsthm}
\usepackage{mathtools}
\usepackage{bm}
\usepackage{dsfont}
\usepackage{graphicx}
\usepackage{caption}
\usepackage{subcaption}
\usepackage{float}
\usepackage{enumitem}
\usepackage{array}
\usepackage{makecell}
\usepackage{multirow}
\usepackage{tabularx}
\usepackage{url}
\usepackage{xcolor}
\usepackage[flushleft]{threeparttable}
\usepackage[ruled,vlined,linesnumbered]{algorithm2e}

\newtheorem{thm}{Theorem}
\newtheorem{lemma}{Lemma}
\theoremstyle{definition}
\newtheorem{definition}{Definition}

\begin{document}

\title{Collaboration-Aware Graph Convolutional Network for Recommender Systems}

\author{Yu Wang}
\email{yu.wang.1@vanderbilt.edu}
\affiliation{%
  \institution{Vanderbilt University}
}

\author{Yuying Zhao}
\email{yuying.zhao@vanderbilt.edu}
\affiliation{%
  \institution{Vanderbilt University}
}

\author{Yi Zhang}
\email{yi.zhang@vanderbilt.edu}
\affiliation{%
  \institution{Vanderbilt University}
}

\author{Tyler Derr}
\email{derr.tyler@vanderbilt.edu}
\affiliation{%
  \institution{Vanderbilt University}
}

\renewcommand{\shortauthors}{Wang, et al.}

\begin{abstract}
Graph Neural Networks (GNNs) have been successfully adopted in recommender systems by virtue of the message-passing that implicitly captures collaborative effect. Nevertheless, most of the existing message-passing mechanisms for recommendation are directly inherited from GNNs without scrutinizing whether the captured collaborative effect would benefit the prediction of user preferences. In this paper, we first analyze how message-passing captures the collaborative effect and propose a recommendation-oriented topological metric, Common Interacted Ratio (CIR), which measures the level of interaction between a specific neighbor of a node with the rest of its neighbors. After demonstrating the benefits of leveraging collaborations from neighbors with higher CIR, we propose a recommendation-tailored GNN, Collaboration-Aware Graph Convolutional Network (CAGCN), that goes beyond 1-Weisfeiler-Lehman(1-WL) test in distinguishing non-bipartite-subgraph-isomorphic graphs. Experiments on six benchmark datasets show that the best CAGCN variant outperforms the most representative GNN-based recommendation model, LightGCN, by nearly 10\% in Recall@20 and also achieves around 80\% speedup. Our code/supplementary is at \href{https://github.com/YuWVandy/CAGCN} {\textcolor{blue}{https://github.com/YuWVandy/CAGCN.}} 
\vspace{-1.5ex}
\end{abstract}

\keywords{\vspace{-3.5ex}\\Recommender systems, graph neural networks, collaborative effect\vspace{-0.5ex}}

\maketitle
\vspace{-1ex}
\section{Introduction}\label{sec-introduction}
Recommender systems aim to alleviate information overload by helping users discover items of interest~\cite{covington2016deep,ying2018graph} and have been widely deployed in real-world applications~\cite{shalaby2022m2trec}. Given historical user-item interactions (e.g., click, purchase, review, and rate), the key is to leverage the collaborative effect~\cite{ebesu2018collaborative, he2017neural, ngcf} to predict how likely users will interact with items. A common paradigm for modeling collaborative effect is to first learn embeddings of users/items capable of recovering historical user-item interactions and then perform top-K recommendation based on the pairwise similarity between the learned user/item embeddings.

Since historical user-item interactions can be naturally represented as a bipartite graph with users/items being nodes and interactions being edges~\cite{li2013recommendation, ngcf, lightgcn} and given the unprecedented success of GNNs in learning node representations~\cite{gcn, hu2022detecting, zhuang2022local, mohamed2020social}, recent research has started to leverage GNNs to learn user/item embeddings for the recommendation. Two pioneering works NGCF~\cite{ngcf} and LightGCN~\cite{lightgcn} leverage graph convolutions to aggregate messages from local neighborhoods, which directly injects the collaborative signal into user/item embeddings. More recently, \cite{chen2021structured, wu2021self} explore the robustness and self-supervised learning~\cite{wang2022gnnssl} of graph convolution for recommendation. However, the message-passing mechanisms in all previous recommendation models are directly inherited from GNNs without carefully justifying how collaborative signals are captured and whether the captured collaborative signals would benefit the prediction of user preference. Such ambiguous understanding on how the message-passing captures collaborative signals would pose the risk of learning uninformative or even harmful user/item representations when adopting GNNs in recommendation. For example, \cite{GTN} shows that a large portion of user interactions cannot reflect their actual purchasing behaviors. In this case, blindly passing messages following existing styles of GNNs could capture harmful collaborative signals from these unreliable interactions, which hinders the performance of GNN-based recommender systems.

To avoid collecting noisy or even harmful collaborative signals in message-passing of traditional GNNs, existing work GTN~\cite{GTN} proposes to adaptively propagate user/item embeddings by adjusting the weight of edges based on items' similarity to users' main preferences (i.e., the trend). However, such similarity is computed based on the learned embeddings that still implicitly encode noisy collaborative signals from unreliable user-item interactions. Worse still, calculating edge weights based on user/item embeddings along the training on the fly is computationally prohibitive and hence prevents the model from being deployed in industrial-level recommendations. SGCN~\cite{chen2021structured} attaches the message-passing with a trainable stochastic binary mask to prune noisy edges. However, the unbiased gradient estimator increases the computational load.

Despite the fundamental importance of capturing beneficial collaborative signals, the related studies are still in their infancy. To fill this crucial gap, we aim to demystify the collaborative effect captured by message-passing and develop new insights towards customizing message-passing for recommendations. Furthermore, these insights motivate us to design a recommendation-tailored GNN, Collaboration-Aware Graph Convolutional Network(CAGCN), that passes neighborhood information based on their Common Interacted Ratio (CIR) via the Collaboration-Aware Graph Convolution (CAGC). Our major contributions are listed as follows:

\begin{itemize}[leftmargin=*]
    \item \textbf{Novel Perspective on Collaborative Effect:} We demystify the collaborative effect by analyzing how message-passing helps capture collaborative signals and when the captured collaborative signals are beneficial in computing users' ranking over items.
    
    \item \textbf{Novel Recommendation-tailored Topological Metric:} We then propose a recommendation-tailored topological metric, Common Interacted Ratio (CIR), and demonstrate the capability of CIR to quantify the benefits of the messages from neighborhoods.

    \item \textbf{Novel Convolution beyond 1-WL for Recommendation:} We integrate CIR into message-passing and propose a novel Collaboration-Aware Graph Convolutional Network (CAGCN). Then we prove that it can go beyond 1-WL test in distinguishing non-bipartite-subgraph-isomorphic graphs, show its superiority on real-world datasets including two newly collected ones, and provide an in-depth interpretation of its advantages.
\end{itemize}

\noindent Next, we comprehensively analyze the collaborative effect captured by message-passing and propose CIR to measure whether the captured collaborative effect benefits the prediction of user preferences. 

\vspace{-4ex}
\section{Analysis on Collaborative Effect}\label{sec-analysis}
Let $\mathcal{G} = (\mathcal{V}, \mathcal{E})$ be the user-item bipartite graph, where the node set $\mathcal{V} = \mathcal{U}\cup\mathcal{I}$ includes the user set $\mathcal{U}$ and the item set $\mathcal{I}$. Following previous work~\cite{ngcf, lightgcn, ultragcn}, we only consider the implicit user-item interactions and denote them as edges $\mathcal{E}$ where $e_{pq}$ represents the edge between node $p$ and $q$. The network topology is described by its adjacency matrix $\mathbf{A}\in \{0, 1\}^{|\mathcal{V}|\times |\mathcal{V}|}$, where $\mathbf{A}_{pq} = 1$ when $e_{pq} \in \mathcal{E}$, and $\mathbf{A}_{pq} = 0$ otherwise. Let $\mathcal{N}_p^l$ denote the set of observed neighbors that are exactly $l$-hops away from $p$ and $\mathcal{S}_p = (\mathcal{V}_{\mathcal{S}_p}, \mathcal{E}_{\mathcal{S}_p})$ be the neighborhood subgraph~\cite{beyond} induced in $\mathcal{G}$ by $\widetilde{\mathcal{N}}_p^1 = \mathcal{N}_p^1\cup\{p\}$. We use $\mathscr{P}_{pq}^{l}$ to denote the set of shortest paths of length $l$ between node $p$ and $q$ and denote one of such paths as $P_{pq}^l$. Note that $\mathscr{P}_{pq}^l=\emptyset$ if it is impossible to have a path between $p$ and $q$ of length $l$, e.g., $\mathscr{P}_{11}^1 =\emptyset$ in an acyclic graph. Furthermore, we denote the initial embeddings of users/items as $\mathbf{E}^0\in\mathbb{R}^{(n + m)\times d^0}$ where $\mathbf{e}_p^0 = \mathbf{E}^0_{p}$ and $d_p$ are the node $p$'s embedding and degree.

Following~\cite{lightgcn, ngcf}, each node has no semantic features but purely learnable embeddings. Therefore, we remove the nonlinear transformation by leveraging LightGCN~\cite{lightgcn} as the canonical architecture and exclusively explore the collaborative effect captured by message-passing. LightGCN passes messages from user $u$/item $i$'s neighbors within $L$-hops to $u/i$:
 \begin{equation}\label{eq-mp2}
 \small
    \mathbf{e}_{u}^{l + 1} = d_{u}^{-0.5}\sum_{j\in \mathcal{N}_u^1} d_{j}^{-0.5}\mathbf{e}_j^l, 
    \mathbf{e}_{i}^{l + 1} = d_{i}^{-0.5}\sum_{v\in \mathcal{N}_i^1} d_v^{-0.5}\mathbf{e}_v ^l, 
\end{equation}
$\forall l\in\{0, ..., L\}$. The propagated embeddings at all layers including the original embedding are aggregated together via mean-pooling:
 \begin{equation}\label{eq-mp}
 \small
    \mathbf{e}_{u} = \frac{1}{(L + 1)}\sum_{l = 0}^{L}\mathbf{e}_u^l, ~~~~~\mathbf{e}_{i} = \frac{1}{(L + 1)}\sum_{l = 0}^{L}\mathbf{e}_i^l, \forall u\in\mathcal{U}, \forall i\in\mathcal{I}
\end{equation}
In the training stage, for each observed user-item interaction $(u,i)$, LightGCN randomly samples a negative item $i^{-}$ that $u$ has never interacted with before, and forms the triple $(u, i, i^{-})$, which collectively forms the set of observed training triples $\mathcal{O}$. After that, the ranking scores of the user over these two items are computed as $y_{ui} = \mathbf{e}_{u}^{\top}\mathbf{e}_{i}$ and $y_{ui^{-}} = \mathbf{e}_{u}^{\top}\mathbf{e}_{i^{-}}$, which are finally used in optimizing the pairwise Bayesian Personalized Ranking (BPR) loss~\cite{bpr}:
\begin{equation}\label{eq-BPR}
\small
 \mathcal{L}_{\text{BPR}} = \sum_{(u, i, i^{-})\in\mathcal{O}}-\ln\sigma(y_{ui} - y_{ui^-}),
\end{equation}
where $\sigma(\cdot)$ is the Sigmoid function, and we omit the $L_2$ regularization here since it is mainly for alleviating overfitting and has no influence on the collaborative effect captured by message passing.

Under the above LightGCN framework, we expect to answer the following two questions:
\begin{itemize}[leftmargin=*]
    \item \boldsymbol{$Q_1$}: How does message-passing capture the collaborative effect and leverage it in computing users' ranking?
    \item \boldsymbol{$Q_2$}: When do collaborations captured by message-passing benefit the computation of users' ranking over items?
\end{itemize}
Next, We address \boldsymbol{$Q_1$} by theoretically deriving users' ranking over items under the message-passing framework of LightGCN and address \boldsymbol{$Q_2$} by proposing the Common Interacted Ratio (CIR) to measure the benefits of leveraging collaborations from each neighbor in computing users' ranking. The answers to the above two questions further motivate our design of Collaboration-Aware Graph Convolutional Network in Section~\ref{sec-framework}.

\begin{figure*}[t!]
    \centering
    \includegraphics[width=1\textwidth]{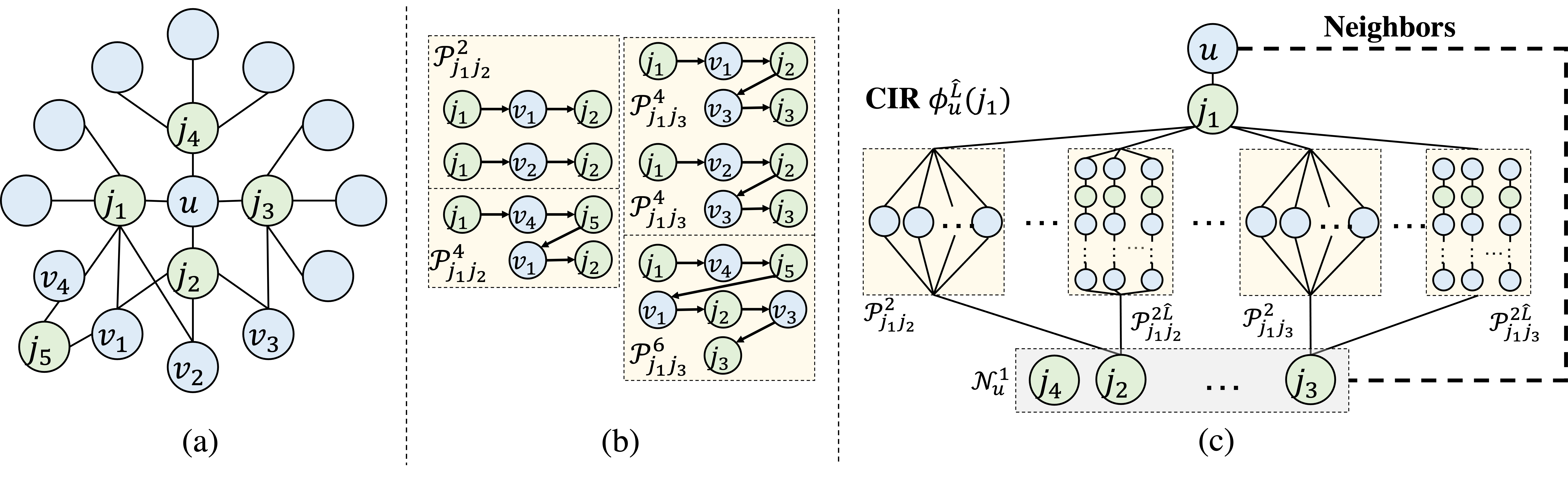}
    \vspace{-6ex}
    \caption{In (a)-(b), $j_1$ has more interactions (paths) with (to) $u$'s neighborhood than $j_4$ and hence is more representative of $u$'s purchasing behaviors than $j_4$. In (c), we quantify CIR between $j_1$ and $u$ via the paths (and associated nodes) between $j_1$ and $\mathcal{N}_u^1$.}
     \label{fig-analysis}
     \vskip -2ex
\end{figure*}

\vspace{-1ex}
\subsection{How does message-passing capture collaborative effect?}
The collaborative effect occurs when the prediction of a user's preference relies on other users' preferences or items' properties~\cite{ricci2011introduction}. Therefore, to answer \boldsymbol{$Q_1$}, we need to seek whether we leverage other nodes' embeddings in computing a specific user's ranking over items. In the inference stage of LightGCN, we take the inner product between user $u$'s embedding and item $i$'s embedding after $L$-layers' message-passing to compute the ranking as\footnote{Detailed derivation is attached in Appendix~\ref{sec-proof1}.}:
\vspace{-1ex}
\begin{equation}\label{eq-collaborative}
y_{ui}^L = (\sum_{l_1=0}^{L}\sum_{j\in\mathcal{N}^{l_1}_{u}}\sum_{l_2 = l_1}^{L}\beta_{l_2}\alpha_{ju}^{l_2}\mathbf{e}_{j}^0)^{\top}(\sum_{l_1=0}^{L}\sum_{v\in\mathcal{N}^{l_1}_{i}}\sum_{l_2 = l_1}^{L}\beta_{l_2}\alpha_{vi}^{l_2}\mathbf{e}_{v}^0),
\end{equation}
\vspace{-2ex}

\noindent where $\alpha_{ju}^{l_2} = \sum_{P_{ju}^{l_2}\in\mathscr{P}_{ju}^{l_2}}\prod_{e_{pq}\in P_{ju}^{l_2}}{d^{-0.5}_{p}d^{-0.5}_{q}}$($\alpha_{ju}^{l_2} = 0$ if $\mathscr{P}_{ju}^{l_2}=\emptyset$) denotes the total weight of all paths of length $l_2$ from $j$ to $u$, $\mathcal{N}^{0}_u = \{u\}$ and specifically, $\alpha_{uu}^{0} = 1$. $\beta_{l_2}$ is the weight measuring contributions of propagated embeddings at layer $l_2$. Thus, based on Eq.~\eqref{eq-collaborative}, we present the answer to \boldsymbol{$Q_1$} as  \boldsymbol{$A_1$}:
\textit{$L$-layer LightGCN-based message-passing captures collaborations between pairs of nodes $\{(j, v)|j\in \bigcup_{l=0}^L\mathcal{N}_u^{l}, v\in\bigcup_{l=0}^L\mathcal{N}_i^{l}\}$, and the collaborative strength of each pair is determined by 1) ${\mathbf{e}_j^0}^{\top}\mathbf{e}_v^0$: embedding similarity between $j$ and $v$, 2) $\{\alpha_{ju}^{l}\}_{l = 0}^L(\{\alpha_{vi}^{l}\}_{l = 0}^L)$: weight of all paths of length $l$ to $L$ from $j$ to $u$ ($v$ to $i$), and 3) $\{\beta_l\}_{l = 0}^L$: the weight of each layer.}

\subsection{When is the captured collaborative effect beneficial to users' ranking?}\label{sec-CIRbenefit}

Although users could leverage collaborations from other users/items as demonstrated above, we cannot guarantee all of these collaborations benefit the prediction of their preferences. For example, in Figure~\ref{fig-analysis}(a)-(b), $u$'s interacted item $j_1$ has more interactions (paths) to $u$'s neighborhoods than $j_4$ and hence is more representative of $u$'s purchasing behaviors~\cite{GTN, chen2021structured}. For each user $u$, we propose the Common Interacted Ratio to quantify the level of interaction between each specific neighbor of $u$ and $u$'s whole item neighborhood:
\begin{definition}
\vspace{-1ex}
\textbf{Common Interacted Ratio (CIR):} For any item $j\in\mathcal{N}_u^1$ of user $u$, the CIR of $j$ around $u$ considering nodes up to $(\widehat{L} + 1)$-hops away from $u$, i.e., $\phi^{\widehat{L}}_u(j)$, is defined as the average interacted ratio of $j$ with all neighboring items of $u$ in $\mathcal{N}_u^1$ through paths of length $\le 2\widehat{L}$:
\end{definition}
\vspace{-4ex}
\begin{equation}\label{eq-phi-test}
\small
    \phi^{\widehat{L}}_u(j) = \frac{1}{|\mathcal{N}_u^1|}\sum_{i\in\mathcal{N}_u^1}\sum_{l = 1}^{\widehat{L}}\alpha^{2l}\sum_{P_{ji}^{2l}\in\mathscr{P}_{ji}^{2l}}{\frac{1}{f(\{\mathcal{N}_k^1|k\in P_{ji}^{2l}\})}}, 
\end{equation}
\vspace{-2ex}

\noindent $\forall j\in\mathcal{N}_u^1, \forall u\in\mathcal{U},$ where $\{\mathcal{N}_k^1|k\in P_{ji}^{2l}\}$ represents the set of the $1$-hop neighborhood of node $k$ along the path $P_{ji}^{2l}$ from node $j$ to $i$ of length $2l$ including $i, j$. $f$ is a normalization function to differentiate the importance of different paths in $\mathscr{P}_{ji}^{2l}$ and its value depends on the neighborhood of each node along the path $P_{ji}^{2l}$. $\alpha^{2l}$ is the importance of paths of length $2l$. 

As shown in Figure~\ref{fig-analysis}(c), $\phi_u^{\widehat{L}}(j_1)$ is decided by paths of length between $2$ to $2\widehat{L}$. By configuring different $\widehat{L}$ and $f$, $\sum_{P_{ji}^{2l}\in\mathscr{P}_{ji}^{2l}}{\frac{1}{f(\{\mathcal{N}_k^1|k\in P_{ji}^{2l}\})}}$ could express many 
graph similarity metrics~\cite{leicht2006vertex, zhou2009predicting, newman2001clustering, salton1989automatic, liben2007link} and we discuss them in Appendix~\ref{sec-toponote}. For simplicity, henceforth we denote $\phi^{\widehat{L}}_u(j)$ as $ \phi_u(j)$. We next empirically verify the importance of leveraging collaborations from neighbors with higher CIR by incrementally adding edges into an initially edge-less graph according to their CIR and visualizing the performance change. Specifically, we consider the performance change in two settings, retraining and pretraining, which are visualized in Figure~\ref{fig-cir-retrain} and ~\ref{fig-cir-pretrain}, respectively. In both of these two settings, we iteratively cycle each node and add its corresponding neighbor according to the CIR until hitting the budget. Here we consider variants of CIR that we later define in Section~\ref{sec-experimentsetting} with further details in Appendix~\ref{sec-toponote}.

\begin{figure}[t!]
     \centering
     \vskip -1ex
     \hspace{-2.5ex}
     \includegraphics[width=0.5\textwidth]{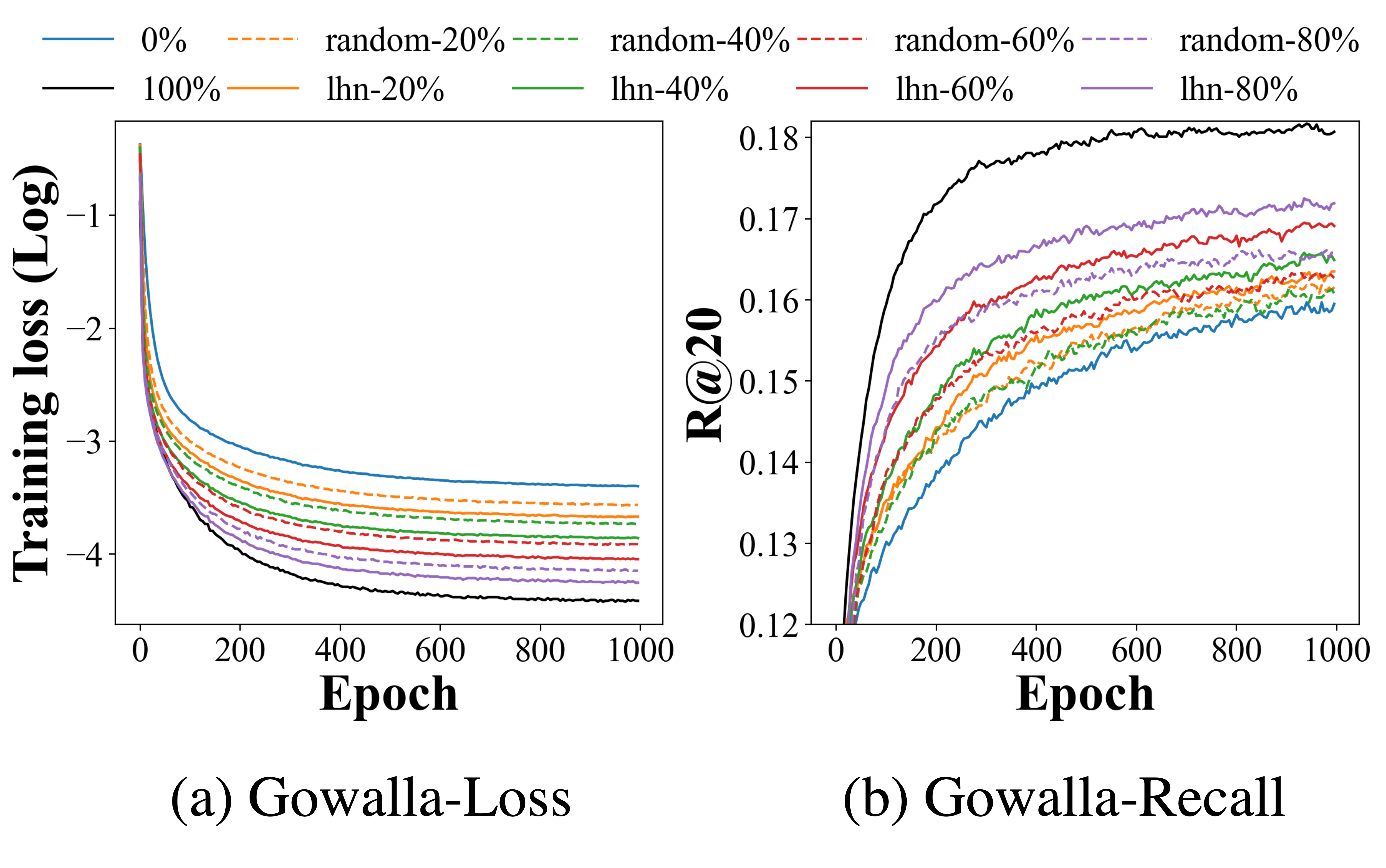}
     \vskip -3ex
     \caption{The training loss (left) is lower and the performance (right) is higher when adding edges according to the variant CIR-lhn (Leicht Holme Nerman) than adding randomly under the same addition budget. Detailed experimental settings and more results are provided in Appendix~\ref{app-res-construct}.}
     \label{fig-cir-retrain}
     \vspace{-4ex}
\end{figure}

For the re-training setting, we first remove all observed edges in the training set to create the edgeless bipartite graph and then incrementally add edges according to their CIR and retrain user/item embeddings. In Figure~\ref{fig-cir-retrain}, we evaluate the performance on the newly constructed bipartite graph under different edge budgets. Clearly, the training loss/performance becomes lower/higher when adding more edges because message-passing captures more collaborative effects. Furthermore, since edges with higher CIR connect neighbors with more connections to the whole neighborhood, optimizing embeddings of nodes incident to these edges pull the whole neighborhood closer and hence leads to the lower training loss over neighborhoods' connections, which causes the overall lower training loss in Figure~\ref{fig-cir-retrain}(a). In Figure~\ref{fig-cir-retrain}(b), we observe that under the same adding budget, adding according to CIRs achieves higher performance than adding randomly. It is because neighbors with higher interactions with the whole neighborhood are more likely to have higher interactions with neighbors to be predicted (We empirically verify this in Table~\ref{tab-simi}.). Then for each user, maximizing its embedding similarity to its training neighbors with higher CIR will indirectly improve its similarity to its to-be-predicted neighbors, which leads to lower population risk and higher generalization/performance.

\begin{figure}[htbp!]
     \centering
     \includegraphics[width=0.5\textwidth]{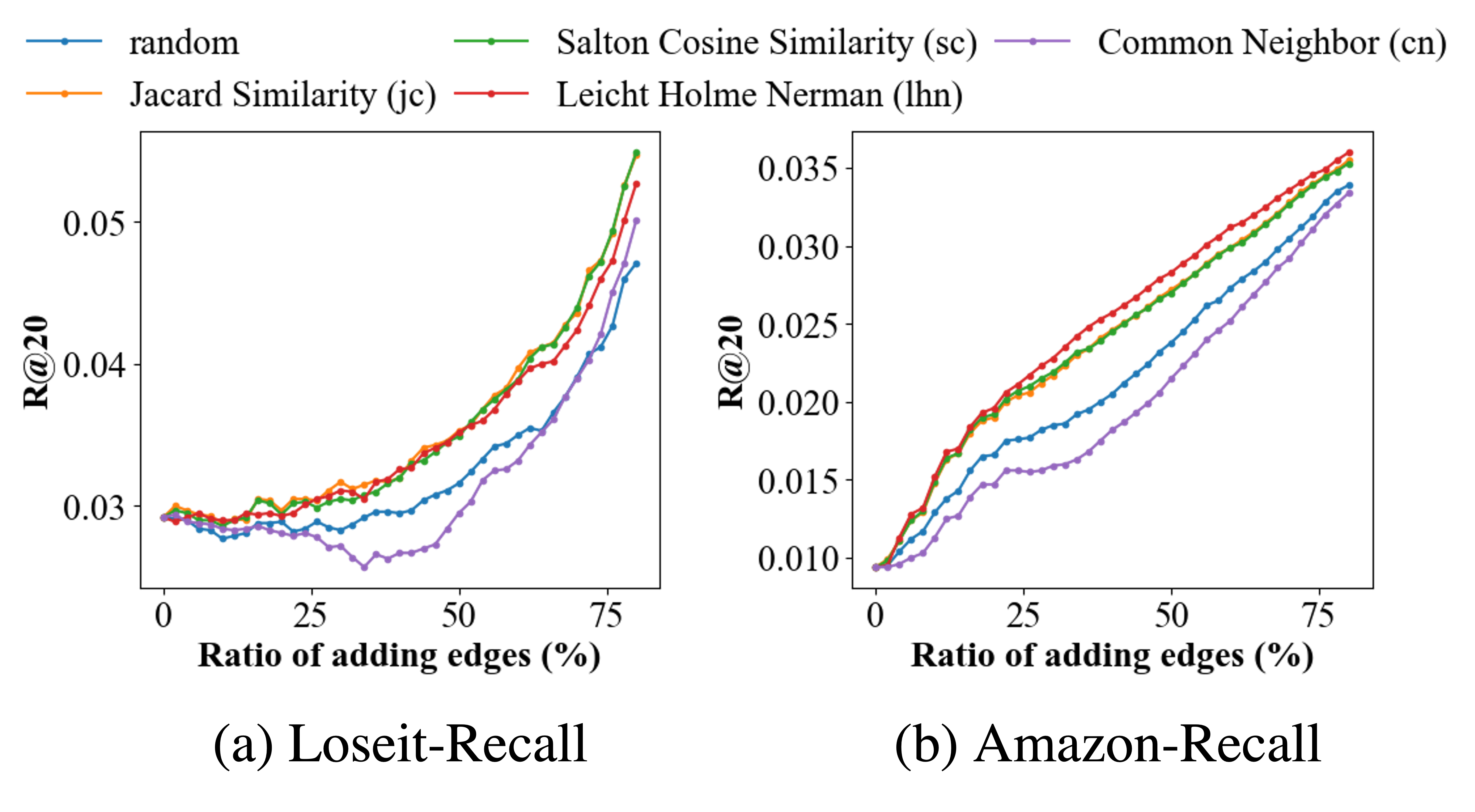}
     \vskip -4ex
     \caption{The performance of adding edges according to CIR variants generally increases faster than adding randomly after pre-training. 
     See Appendix~\ref{app-res-construct2} for more results.}
     \label{fig-cir-pretrain}
     \vskip -1ex
\end{figure}

For the pre-training setting, we first pre-train user/item embeddings on the original bipartite graph and then propagate the pre-trained embeddings on the newly constructed bipartite graph under different edge budgets. This setting is more realistic since in the real world, with the exponential interactions streamingly coming in~\cite{wang2020streaming} while the storage space is limited, we are forced to keep only partial interactions and the pre-trained user/item embeddings. Figure~\ref{fig-cir-pretrain} demonstrates that under the same adding budget, keeping edges according to CIR leads to higher performance than keeping randomly, which further verifies the effectiveness of CIR in quantifying the edge importance. An interesting observation is that adding more edges cannot always bring performance gain as shown in Figure~\ref{fig-cir-pretrain}(a) when the ratio of added edges is between 0\%-20\%. We hypothesize there are two reasons. From network topology, only when edges are beyond a certain level can the network form a giant component so that users could receive enough neighborhood information. Secondly, from representation learning, more nodes would have inconsistent neighborhood contexts between the training and the inference when only a few edges are added. Such inconsistent neighborhood context would compromise the performance and will be alleviated when more edges are added as shown later in Figure~\ref{fig-cir-pretrain}(a). Furthermore, different CIR variants cause different increasing speeds of performance. For example, sc is faster on Loseit in Figure~\ref{fig-cir-pretrain}(a) while lhn is faster on Amazon in Figure~\ref{fig-cir-pretrain}(b). Except for the cn, jc/sc/lhn lead to faster improvement than the random one, which highlights the potential of CIR in devising cost-effective strategies for pruning edges in the continual learning~\cite{wang2022lifelong}.

From the above analysis, we summarize the answer \boldsymbol{$A_2$} to \boldsymbol{$Q_2$} as: \textit{Leveraging collaborations from $u$'s neighboring node $j$ with higher CIR $\phi_u(j)$ would cause more benefits to $u$'s ranking.}

\section{Collaboration-Aware Graph Convolutional Networks}\label{sec-framework}

The former section demonstrates that passing messages according to neighbors' CIR is crucial in improving users' ranking. This motivates us to propose a new graph convolution operation, Collaboration-Aware Graph Convolution(CAGC), which passes node messages based on the benefits of their provided collaborations. Furthermore, we wrap the proposed CAGC within LightGCN and develop two CAGC-based models.

\subsection{Collaboration-Aware Graph Convolution}
The core idea of CAGC is to strengthen/weaken the messages passed from neighbors with higher/lower CIR to center nodes. To achieve this, we compute the edge weight as:
\begin{equation}\label{eq-sign}
    {\boldsymbol{\Phi}}_{ij} = \begin{cases}
    \phi_i(j), & \text{if $\mathbf{A}_{ij} > 0$}\\
    0, & \text{if $\mathbf{A}_{ij} = 0$}
    \end{cases}, \forall i, j \in \mathcal{V}
\end{equation}
where $\phi_{i}(j)$ is the CIR of neighboring node $j$ centering around $i$. Note that unlike the symmetric graph convolution $\mathbf{D}^{-0.5}\mathbf{A}\mathbf{D}^{-0.5}$ used in LightGCN, here $\boldsymbol{\Phi}$ is unsymmetric. This is rather interpretable: the interacting level of node $j$ with $i$'s neighborhood is likely to be different from the interacting level of node $i$ with $j$'s neighborhood. We further normalize $\boldsymbol{\Phi}$ and combine it with the LightGCN convolution:
\begin{equation}\label{eq-CAGCNaug}
    \mathbf{e}_i^{l+1} = \sum_{j\in\mathcal{N}_i^1}g(\gamma_i\frac{\boldsymbol{\Phi}_{ij}}{\sum_{k\in\mathcal{N}_i^1}{\boldsymbol{\Phi}_{ik}}}, d_i^{-0.5}d_j^{-0.5})\mathbf{e}_j^l, \forall i\in\mathcal{V}
\end{equation}

\noindent where $\gamma_i$ is a coefficient that varies the total amount of messages flowing to node $i$ and controls its embedding magnitude~\cite{park2020trap}. $g$ is a function combining the edge weights computed based on CIR and LightGCN. We could either simply set $g$ as the weighted summation of these two propagated embeddings or learn $g$ by parametrization. Next, we prove that for certain choices of $g$, CAGC can go beyond 1-WL in distinguishing non-bipartite-subgraph-isomorphic graphs. First, we prove the equivalence between the subtree-isomorphism and the subgraph-isomorphism in bipartite graphs:
\begin{thm}\label{thm-eq}
    In bipartite graphs, two subgraphs that are subtree-isomorphic if and only if they are subgraph-isomorphic\footnote{Definitions of subtree-/subgraph-isomorphism are in \textcolor{blue}{Supplementary}~\ref{app-graphiso}\cite{beyond}.}.
\end{thm}

\begin{proof}
\vspace{-1.5ex}
We prove this theorem in two directions. Firstly ($\Longrightarrow$), we prove that in a bipartite graph, two subgraphs that are subtree-isomorphic are also subgraph-isomorphic by contradiction. Assuming that there exists two subgraphs $\mathcal{S}_u$ and $\mathcal{S}_i$ that are subtree-isomorphic yet not subgraph-isomorphic in a bipartite graph, i.e., $\mathcal{S}_u\cong_{subtree} \mathcal{S}_i$ and $\mathcal{S}_u \not\cong_{subgraph} \mathcal{S}_i$. By definition of subtree-isomorphism, we trivially have $\mathbf{e}_v^l=\mathbf{e}_{h(v)}^l, \forall v\in\mathcal{V}_{\mathcal{S}_u}$. Then to guarantee $\mathcal{S}_u \not\cong_{subgraph} \mathcal{S}_i$ and also since edges are only allowed to connect $u$ and its neighbors $\mathcal{N}_u^1$ in the bipartite graph, there must exist at least an edge $e_{uv}$ between $u$ and one of its neighbors $v\in\mathcal{N}_u^1$ such that $e_{uv}\in\mathcal{E}_{\mathcal{S}_u}, e_{h(u)h(v)}\notin \mathcal{E}_{\mathcal{S}_i}$, which contradicts the assumption that $\mathcal{S}_u\cong_{subtree} \mathcal{S}_i$. Secondly ($\Longleftarrow$), we can prove that in a bipartite graph, two subgraphs that are subgraph-isomorphic are also subtree-isomorphic, which trivially holds since in any graph, subgraph-isomorphism leads to subtree-isomorphism~\cite{beyond}.
\end{proof}

Since 1-WL test can distinguish subtree-isomorphic graphs~\cite{beyond}, the equivalence between these two isomorphisms indicates that in bipartite graphs, both of the subtree-isomorphic graphs and subgraph-isomorphic graphs can be distinguished by 1-WL test. Therefore, to go beyond 1-WL in bipartite graphs, we need to propose a novel graph isomorphism, bipartite-subgraph-isomorphism in Definition~\ref{df-bisub}, which is even harder to be distinguished than the subgraph-isomorphism by 1-WL test.

\begin{definition}\label{df-bisub}
\textbf{Bipartite-subgraph-isomorphism:} 
$\mathcal{S}_u$ and $\mathcal{S}_i$ are bipartite-subgraph-isomorphic, denoted as $\mathcal{S}_u\cong_{bi-subgraph} \mathcal{S}_i$, if there exists a bijective mapping $h:\widetilde{\mathcal{N}}^1_{u}\cup\mathcal{N}^2_{u} \rightarrow \widetilde{\mathcal{N}}^1_{i}\cup\mathcal{N}^2_{i}$ such that $h(u) = i$ and $\forall v, v'\in\widetilde{\mathcal{N}}_u^1\cup\mathcal{N}_u^2$, $e_{vv'}\in\mathcal{E} \iff e_{h(v)h(v')}\in\mathcal{E}$ and $\mathbf{e}^l_{v} = \mathbf{e}^l_{h(v)}, \mathbf{e}^l_{v'} = \mathbf{e}^l_{h(v')}$.
\end{definition}

\begin{lemma}\label{lemma-injective}
    If $g$ is multilayer perceptron (MLP), then we have that $g(\{(\gamma_i\widetilde{\boldsymbol{\Phi}}_{ij}, \mathbf{e}_j^l)|j\in\mathcal{N}_i^1\}, \{(d_i^{-0.5}d_j^{-0.5}, \mathbf{e}_j^l)|j\in\mathcal{N}_i^1\})$ is injective.
\end{lemma}

\begin{proof}
If we assume that all node embeddings share the same discretization precision, then embeddings of all nodes in a graph can form a countable set $\mathcal{H}$. Similarly, for each edge in a graph, its CIR-based weight $\widetilde{\boldsymbol{\Phi}}_{ij}$ and degree-based weight $d_i^{-0.5}d_j^{-0.5}$ can also form two different countable sets $\mathcal{W}_1, \mathcal{W}_2$ with $|\mathcal{W}_1| = |\mathcal{W}_2|$. Then $\mathcal{P}_1 = \{\widetilde{\boldsymbol{\Phi}}_{ij}\mathbf{e}_i|\widetilde{\boldsymbol{\Phi}}_{ij}\in \mathcal{W}_1, \mathbf{e}_i\in \mathcal{H}\}, \mathcal{P}_2 = \{d_i^{-0.5}d_j^{-0.5}\mathbf{e}_i|d_i^{-0.5}d_j^{-0.5}\in \mathcal{W}_2, \mathbf{e}_i\in \mathcal{H}\}$ are also two countable sets. Let $P_1, P_2$ be two multisets containing elements from $\mathcal{P}_1$ and $\mathcal{P}_2$, respectively, and $|P_1| = |P_2|$. Then by Lemma 1 in \cite{beyond}, there exists a function $s$ such that $\pi(P_1, P_2) = \sum_{p_1 \in P_1, p_2 \in P_2}s(p_1, p_2)$ is unique for any distinct pair of multisets $(P_1, P_2)$. Since the MLP-based g is a universal approximator~\cite{xu2018powerful} and hence can learn $s$, we know that $g$ is injective.
\end{proof}

\begin{thm}\label{thm-1WL}
    Let M be a GNN with sufficient number of CAGC-based convolution layers defined by Eq.~\eqref{eq-CAGCNaug}. If $g$ is MLP, then M is strictly more expressive than 1-WL in distinguishing subtree-isomorphic yet non-bipartite-subgraph-isomorphic graphs. 
\end{thm}

\begin{proof}
We prove this theorem in two directions. Firstly ($\Longrightarrow$), following~\cite{beyond}, we prove that the designed CAGC here can distinguish any two graphs that are distinguishable by 1-WL by contradiction. Assume that there exist two graphs $\mathcal{G}_1$ and $\mathcal{G}_2$ which can be distinguished by 1-WL but cannot be distinguished by CAGC. Further, suppose that 1-WL cannot distinguish these two graphs in the iterations from $0$ to $L-1$, but can distinguish them in the $L^{\text{th}}$ iteration. Then, there must exist two neighborhood subgraphs $\mathcal{S}_u$ and $\mathcal{S}_i$ whose neighboring nodes correspond to two different sets of node labels at the $L^{\text{th}}$ iteration, i.e., $\{\mathbf{e}_v^l|v\in\mathcal{N}_u^1\}\ne\{\mathbf{e}_j^l|j\in\mathcal{N}_i^1\}$. Since $g$ is injective by Lemma~\ref{lemma-injective}, for $\mathcal{S}_u$ and $\mathcal{S}_i$, $g$ would yield two different feature vectors at the $L^{\text{th}}$ iteration.
This means that CAGC can also distinguish $\mathcal{G}_1$ and $\mathcal{G}_2$, which contradicts the assumption. 

Secondly ($\Longleftarrow$), we prove that there exist at least two graphs that can be distinguished by CAGC but cannot be distinguished by 1-WL. Figure~\ref{fig-example} in \textcolor{blue}{Supplementary}~\ref{app-graphiso} presents two of such graphs $\mathcal{S}_u, \mathcal{S}_u'$, which are subgraph isomorphic but non-bipartite-subgraph-isomorphic. Assuming $u$ and $u'$ have exactly the same neighborhood feature vectors $\mathbf{e}$, then directly propagating according to 1-WL or even considering node degree as the edge weight as GCN~\cite{gcn} can still end up with the same propagated feature for $u$ and $u'$. However, if we leverage JC to calculate CIR as introduced in Appendix~\ref{sec-toponote}, then we end up with $\{(d_ud_{j_1})^{-0.5}\mathbf{e}, (d_ud_{j_2})^{-0.5}\mathbf{e}, (d_ud_{j_3})^{-0.5}\mathbf{e}\} \ne \{(d_{u'}^{-0.5}d_{j'_1}^{-0.5} + \boldsymbol{\widetilde{\Phi}}_{u'j'_1})\mathbf{e}, (d_{u'}^{-0.5}d_{j'_2}^{-0.5} + \boldsymbol{\widetilde{\Phi}}_{u'j'_2})\mathbf{e}, (d_{u'}^{-0.5}d_{j'_3}^{-0.5} + \boldsymbol{\widetilde{\Phi}}_{u'j'_3})\mathbf{e}\}$. Since $g$ is injective by Lemma~\ref{lemma-injective}, CAGC would yield two different embeddings for $u$ and $u'$.
\end{proof}

Theorem~\ref{thm-1WL} indicates that GNNs whose aggregation scheme is CAGC can distinguish non-bipartite-subgraph-isomorphic graphs that are indistinguishable by 1-WL.

\subsection{Model Architecture and Complexity}
Following the principle of LightGCN that the designed graph convolution should be light and easy to train, except for the message-passing component, all other components of our proposed CAGC-based models is exactly the same as LightGCN including the average pooling and the model training, which have already been covered in Section~\ref{sec-analysis}. We provide the detailed time/space complexity comparison between our models and all other baselines in Appendix~\ref{sec-complexity}.  We visualize the architecture of CAGC-based models in Figure~\ref{fig-model1}. Based on the choice of $g$, we have two specific model variants. For the first variant CAGCN, we calculate the edge weight solely based on CIR in message-passing by setting $g(A, B) = A$ in Eq.\eqref{eq-CAGCNaug} and set $\gamma_i = \sum_{r \in \mathcal{N}_i^1}d_i^{-0.5}d_r^{-0.5}$ to ensure that the total edge weights for messages received by each node are the same as the one in LightGCN. For CAGCN*, we set $g$ as the weighted summation and set $\gamma_i = \gamma$ as a constant controlling the trade-off between contributions from message-passing by LightGCN and by CAGC. We term the model variant as CAGCN(*)-jc if we use Jaccard Similarity (JC)~\cite{liben2007link} to compute $\boldsymbol{\Phi}$. The same rule applies to other topological metrics listed in Appendix~\ref{sec-toponote}. Concrete equations of CAGCN and CAGCN* are provided in Appendix~\ref{app-variant}.

\begin{figure}[ht!]
     \centering
     \vskip -3ex
     \hspace{-3ex}
     \includegraphics[width=0.495\textwidth]{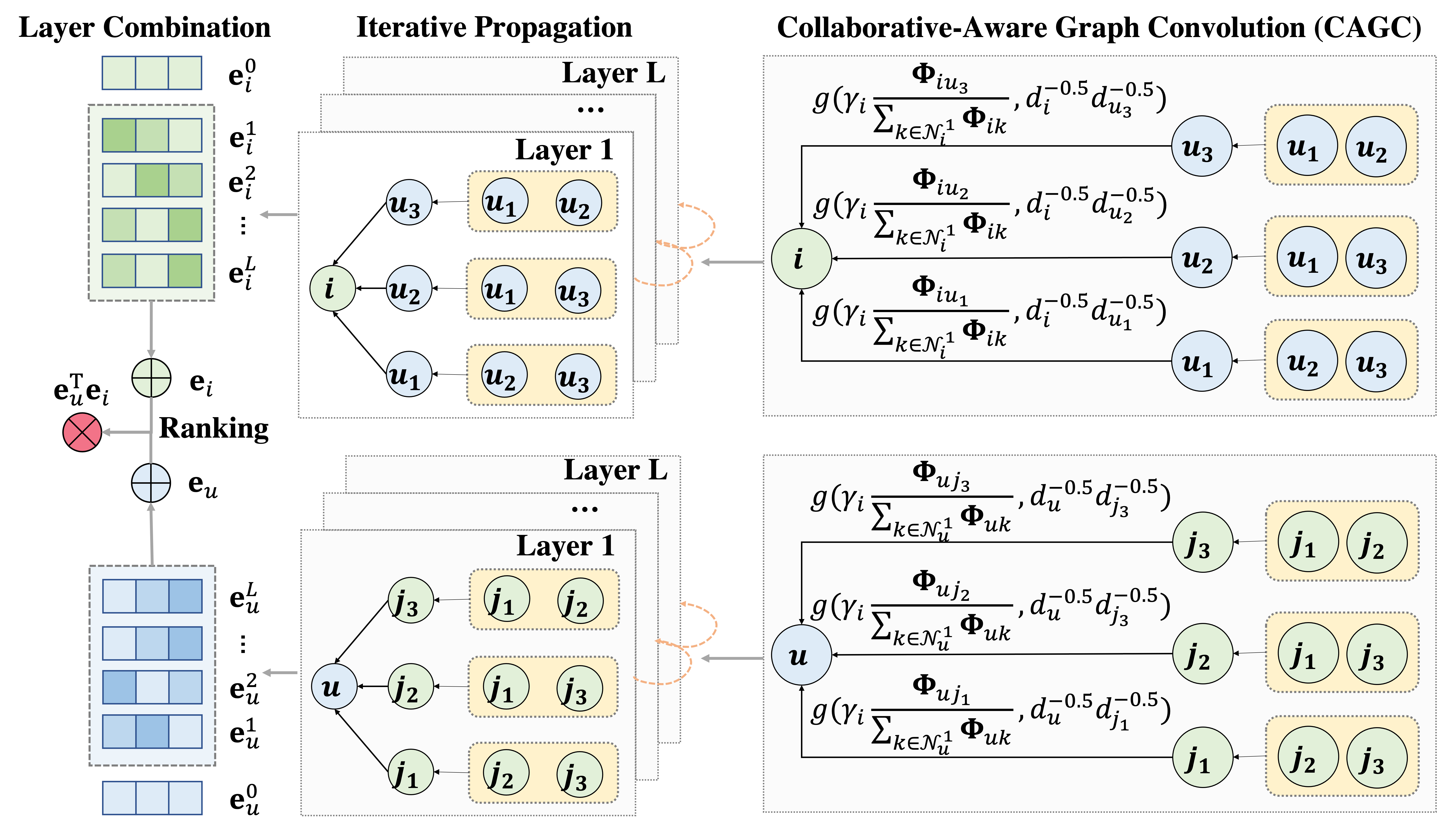}
     \vskip -2ex
     \caption{The architecture of the proposed CAGCN(*).}
     \label{fig-model1}
     \vspace{-4ex}
\end{figure}

\vspace{-1.5ex}
\section{Experiments}\label{sec-experiments}
In this section, we conduct experiments to evaluate CAGCN(*).
\vspace{-2ex}
\subsection{Experimental Settings}\label{sec-experimentsetting}
\subsubsection{Datasets.} Following~\cite{lightgcn, ngcf}, we validate the proposed approach on \textbf{Gowalla}, \textbf{Yelp}, \textbf{Amazon}, and \textbf{Ml-1M}, the details of which are provided in~\cite{lightgcn, ngcf}. Moreover, we collect two extra datasets to further demonstrate the superiority of our proposed model in even broader user-item interaction domains: \textbf{(1) Loseit}: This dataset is collected from subreddit \textit{loseit - Lose the Fat}\footnote{https://www.reddit.com/r/loseit/} from March 2020 to March 2022 where users discuss healthy and sustainable methods of losing weight via posts. To ensure the quality of this dataset, we use the 10-core setting~\cite{he2016vbpr}, i.e., retaining users and posts with at least ten interactions. \textbf{(2) News}: This dataset includes the interactions from subreddit \textit{World News}\footnote{https://www.reddit.com/r/worldnews/} where users share major news around the world via posts. Similarly, we use the 10-core setting to ensure the quality of this dataset. We summarize the statistics of all six datasets in Table~\ref{tab-dataset}.

\begin{table*}[t]
\vspace{-0.5ex}
\small
\setlength{\extrarowheight}{.125pt}
\setlength\tabcolsep{6pt}
\caption{Performance comparison of CAGCN(*) with baselines. The best and runner-up results are in \textbf{bold} and \underline{underlined}.}
\vskip -3ex
\label{tab-resfull}
\begin{tabular}{lc|cccc|cccc|ccc}
\Xhline{2.5\arrayrulewidth}
\multirow{2}{*}{\textbf{Model}} & \multicolumn{1}{c|}{\multirow{2}{*}{\textbf{Metric}}} & \multicolumn{1}{c}{\multirow{2}{*}{\textbf{MF}}} & \multicolumn{1}{c}{\multirow{2}{*}{\textbf{NGCF}}} & \multicolumn{1}{c}{\multirow{2}{*}{\textbf{LightGCN}}} & \multicolumn{1}{c|}{\multirow{2}{*}{\textbf{UltraGCN}}} & \multicolumn{4}{c|}{\textbf{CAGCN}} & \multicolumn{3}{c}{\textbf{CAGCN*}} \\
 & \multicolumn{1}{c|}{} & \multicolumn{1}{c}{} & \multicolumn{1}{c}{} & \multicolumn{1}{c}{} & \multicolumn{1}{c|}{} & \textbf{-jc} & \textbf{-sc} & \textbf{-cn} & \textbf{-lhn} & \textbf{-jc} & \textbf{-sc} & \textbf{-lhn} \\
\Xhline{2.5\arrayrulewidth}
\multirow{2}{*}{Gowalla} & Recall@20 & 0.1554 & 0.1563 & 0.1817 &\underline{0.1867} & 0.1825 & 0.1826 & 0.1632 & 0.1821 & \textbf{0.1878} & \textbf{0.1878} & 0.1857 \\
 & NDCG@20 & 0.1301 & 0.1300 & 0.1570 & 0.1580 & 0.1575 & 0.1577 & 0.1381 & 0.1577 & \textbf{0.1591} & \underline{0.1588} & 0.1563 \\
\hline
\multirow{2}{*}{Yelp2018} & Recall@20 & 0.0539 & 0.0596 & 0.0659 & 0.0675 & 0.0674 & 0.0671 & 0.0661 & 0.0661 & \underline{0.0708} & \textbf{0.0711} & 0.0676 \\
 & NDCG@20 & 0.0460 & 0.0489 & 0.0554 & 0.0553 & 0.0564 & 0.0560 & 0.0546 & 0.0555 & \underline{0.0586} & \textbf{0.0590} & 0.0554 \\
\hline
\multirow{2}{*}{Amazon} & Recall@20 & 0.0337 & 0.0336 & 0.0420 & \textbf{0.0682} & 0.0435 & 0.0435 & 0.0403 & 0.0422 & \underline{0.0510} & 0.0506 & 0.0457 \\
 & NDCG@20 & 0.0265 & 0.0262 & 0.0331 & \textbf{0.0553} & 0.0343 & 0.0342 & 0.0321 & 0.0333 & \underline{0.0403} & 0.0400 & 0.0361 \\
\hline
\multirow{2}{*}{Ml-1M} & Recall@20 & 0.2604 & 0.2619 & 0.2752 & 0.2783 & 0.2780 & 0.2786 & 0.2730 & 0.2760 & \underline{0.2822} & \textbf{0.2827} & 0.2799 \\
 & NDCG@20 & 0.2697 & 0.2729 & 0.2820 & 0.2638 & \underline{0.2871} & \textbf{0.2881} & 0.2818 & \underline{0.2871} & 0.2775 & 0.2776 & 0.2745 \\
\hline
\multirow{2}{*}{Loseit} & Recall@20 & 0.0539 & 0.0574 & 0.0588 & 0.0621 & 0.0622 & 0.0625 & 0.0502 & 0.0592 & \underline{0.0654} & \textbf{0.0658} & \textbf{0.0658} \\
 & NDCG@20 & 0.0420 & 0.0442 & 0.0465 & 0.0446 & 0.0474 & 0.0470 & 0.0379 & 0.0461 & \underline{0.0486} & 0.0484 & \textbf{0.0489} \\
\hline
\multirow{2}{*}{News} & Recall@20 & 0.1942 & 0.1994 & 0.2035 & 0.2034 & 0.2135 & 0.2132 & 0.1726 & 0.2084 & \textbf{0.2182} & \underline{0.2172} & 0.2053 \\
 & NDCG@20 & 0.1235 & 0.1291 & 0.1311 & 0.1301 & 0.1385 & 0.1384 & 0.1064 & 0.1327 & \underline{0.1405} & \textbf{0.1414} & 0.1311 \\
\Xhline{2.5\arrayrulewidth}
\multicolumn{1}{r}{\multirow{2}{*}{\textbf{Avg. Rank}}} & Recall@20 & 9.83 & 9.17 & 7.33 & 4.17 & 4.67 & 4.33 & 8.83 & 6.17 & \underline{1.67} & \textbf{1.50} & 3.33  \\
 & NDCG@20 & 9.50 & 9.17 & 5.83 & 6.00 &  \underline{3.67} & 4.00 &  8.33 & 5.00 &  \textbf{2.50} & \textbf{2.50} & 5.17 \\
\Xhline{2.5\arrayrulewidth}
\end{tabular}

\begin{tablenotes}
      \small
      \centering
      \item \textbf{jc}-Jacard Similarity, \textbf{sc}-Salton Cosine Similarity, \textbf{cn}-Common Neighbors, \textbf{lhn}-Leicht-Holme-Nerman
\end{tablenotes}

\vskip -1ex
\vspace{-1ex}

\end{table*}

\begin{table}[htbp!]
\normalsize
\caption{Basic dataset statistics.}
\centering
\vspace{-3ex}
\begin{center}
\begin{tabular}{lcccc}
\Xhline{2\arrayrulewidth}
\textbf{Dataset} & \# \textbf{Users} & \# \textbf{Items} & \# \textbf{Interactions} & \textbf{Density} \\
\Xhline{2\arrayrulewidth}
Gowalla & 29, 858 & 40, 981 & 1, 027, 370 & 0.084\% \\
Yelp & 31, 668 & 38, 048 & 1, 561, 406 & 0.130\% \\
Amazon & 52, 643 & 91, 599 & 2, 984, 108 & 0.062\% \\
Ml-1M & 6, 022 & 3, 043 & 895, 699 & 4.888\% \\
Loseit & 5, 334 & 54, 595 & 230, 866 & 0.08\%\\
News & 29, 785 & 21, 549 & 766, 874 & 0.119\%\\
\Xhline{2\arrayrulewidth}
\end{tabular}
\end{center}
\vskip -1ex
\begin{tablenotes}
  \small
  \item \hspace{2ex} 
  \textbf{*}\textbf{Yelp}: Yelp2018; \textbf{*}\textbf{Amazon}: Amazon-Books;\textbf{*}\textbf{Ml-1M}: Movielens-1M.
\end{tablenotes}
\label{tab-dataset}
\vskip -3ex
\end{table}

\vspace{-2.5ex}
\subsubsection{Baseline methods.} We compare our model with MF, NGCF, LightGCN, UltraGCN, GTN~\cite{bpr, ngcf, lightgcn, ultragcn, GTN}. Details of them are clarified in Appendix~\ref{app-baseline}. Since here the purpose is to evaluate the effectiveness of CAGC-based message-passing, we only compare with baselines that focus on graph convolution (besides the classic MF) including the state-of-the-art GNN-based recommendation models (i.e., UltraGCN and GTN). Note that our work could be further enhanced if incorporating other techniques such as contrastive learning to derive self-supervision but stacking these would sidetrack the main topic of this paper, graph convolution, so we leave them as one future direction.


\vspace{-1.5ex}
\subsubsection{Evaluation Metrics}
Two popular metrics: Recall and Normalized Discounted Cumulative Gain(NDCG)~\cite{ngcf} are adopted for evaluation. We set the default value of K as 20 and report the average of Recall@20 and NDCG@20 over all users in the test set. During inference, we treat items that the user has never interacted with in the training set as candidate items. All models predict users' preference scores over these candidate items and rank them based on the computed scores to further calculate Recall@20 and NDCG@20.

\vspace{-2ex}
\subsection{Performance Comparison}

We first compare our proposed CAGCN-variants with LightGCN. In Table~\ref{tab-resfull}, CAGCN-jc/sc/lhn achieves higher performance than LightGCN because we aggregate more information from nodes with higher CIR(jc, sc, lhn) that bring more beneficial collaborations as justified in Section~\ref{sec-CIRbenefit}. CAGCN-cn generally performs worse than LightGCN because nodes having more common neighbors with other nodes tend to have higher degrees and blindly aggregating information more from these nodes would cause false-positive link prediction. Since different datasets exhibit different patterns of $2^{\text{nd}}$-order connectivity, there is no fixed topological metric that performs the best among all datasets. For example, CAGCN-jc performs better than CAGCN-sc on Yelp and News, while worse on Gowalla, Ml-1M. 

Then, we compare CAGCN*-variants with other baselines. We omit CAGCN*-cn here due to the worse performance of CAGCN-cn than LightGCN. We can see that CAGCN*-jc/sc almost consistently achieves higher performance than other baselines except for UltraGCN on Amazon. This is because UltraGCN allows multiple negative samples for each positive interaction, e.g., 500 negative samples here on Amazon\footnote{
\href{https://github.com/xue-pai/UltraGCN}
{\textcolor{blue}{UltraGCN}} negative samples: 1500/800/500/200 on Gowalla/Yelp2018/Amazon/Ml-1M.}, which lowers the efficiency as we need to spend more time preparing a large number of negative samples per epoch. Among the baselines, UltraGCN exhibits the strongest performance because it approximates the infinite layers of message passing and constructs the user-user graphs to capture 2$^{\text{nd}}$-order connectivity. LightGCN and NGCF perform better than MF since they inject the collaborative effect directly through message-passing.

\begin{table}[t!]
\setlength{\extrarowheight}{.095pt}
\setlength\tabcolsep{6pt}
\caption{Performance comparison of CAGCN* with GTN.}
\label{tab-gtn}
\vskip -3ex
\centering
\small
\begin{tabular}{lc|c|cccc}
\Xhline{2\arrayrulewidth}
\multirow{2}{*}{\textbf{Model}} & \multicolumn{1}{c|}{\multirow{2}{*}{\textbf{Metric}}} & \multicolumn{1}{c|}{\multirow{2}{*}{\textbf{GTN}}} & \multicolumn{3}{c}{\textbf{CAGCN*}} \\
  &  &  & \textbf{-jc} & \textbf{-sc} & \textbf{-lhn} \\
\Xhline{2\arrayrulewidth}
\multirow{2}{*}{Gowalla} & Recall@20 & 0.1870 &  \textbf{0.1901} & \underline{0.1899} & 0.1885\\
 & NDCG@20 & 0.1588 & \textbf{0.1604} & \underline{0.1603} & 0.1576\\
 \Xhline{2\arrayrulewidth}
\multirow{2}{*}{Yelp2018} & Recall@20 & 0.0679 & \textbf{0.0731} & \underline{0.0729} & 0.0689\\
 & NDCG@20 & 0.0554 & \textbf{0.0605} & \underline{0.0601} & 0.0565 & \\
 \Xhline{2\arrayrulewidth}
\multirow{2}{*}{Amazon} & Recall@20 & 0.0450 & \underline{0.0573} & \textbf{0.0575} & 0.0520\\
 & NDCG@20 & 0.0346 & \underline{0.0456} & \textbf{0.0458} & 0.0409\\
 \Xhline{2\arrayrulewidth}
\end{tabular}
\vskip -2ex
\vspace{-2ex}
\end{table}
To align the setting with GTN, we increase the embedding size $d^0$ to 256 following~\cite{GTN}\footnote{As the user/item embedding is a significant hyperparameter, it is crucial to ensure the same embedding size when comparing models; thus, we separately compare against \href{https://github.com/wenqifan03/GTN-SIGIR2022}{\textcolor{blue}{GTN}} using their larger embedding size.} and observe the consistent superiority of our model over GTN in Table~\ref{tab-gtn}. This is because in GTN~\cite{GTN}, the edge weights for message-passing are still computed based on node embeddings that implicitly encode noisy collaborative signals from unreliable interactions. Conversely, our CAGCN* directly alleviates the propagation on unreliable interactions based on its CIR value, which removes noisy interactions from the source.

\subsection{Efficiency Comparison}\label{sec-effcompare}
As recommendation models will be eventually deployed in user-item data of real-world scale, it is crucial to compare the efficiency of the proposed CAGCN(*) with other baselines. To guarantee a fair comparison, we use a uniform code framework implemented ourselves for all models and run them on the same machine with Ubuntu 20.04 system, AMD Ryzen 9 5900 12-Core Processor (3.0 GHz), 128 GB RAM and GPU NVIDIA GeForce RTX 3090. We report the Recall@20 on Yelp and NDCG@20 on Loseit achieved by the best CAGCN(*) variant based on Table~\ref{tab-resfull}. We track the performance and the training time per 5 epochs. Complete results are included in \textcolor{blue}{Supplementary}~\ref{app-effcompare}. In Figure~\ref{fig-train}(a)-(b), CAGCN achieves higher performance than LightGCN in less time. We hypothesize that for each user, its neighbors with higher interactions with its whole neighborhood would also have higher interactions with its interacted but unobserved neighbors. Then as CAGCN aggregate more information from these observed neighbors that have higher interactions with the whole neighborhood, it indirectly enables the user to aggregate more information from its to-be-predicted neighbors. 

To verify the above hypothesis, we define the to-be-predicted neighborhood set of user $u$ in the testing set as $\widehat{\mathcal{N}}^1_u$ and for each neighbor $j\in\mathcal{N}_u^1$, calculate its CIR $\widehat{\phi}_u^{\widehat{L}}(j)$ with nodes in $\widehat{\mathcal{N}}^1_u$. Then we compare the ranking consistency among CIRs calculated from training neighborhoods (i.e., $\phi_u(j)$), from testing neighborhoods (i.e., $\widehat{\phi}_u(j)$) and from full neighborhoods (we replace $\widehat{\mathcal{N}}_u^1$ with $\mathcal{N}_u^1\cup$ $\widehat{\mathcal{N}}_u^1$ in Eq.~\eqref{eq-phi-test}). Here we respectively use four topological metrics (JC, SC, LHN, and CN) to define $f$ and rank the obtained three lists. Then, we measure the similarity of the ranked lists between Train-Test and between Train-Full by Rank-Biased Overlap (RBO)~\cite{rbo}. The averaged RBO values over all nodes $v\in\mathcal{V}$ on three datasets are shown in Table~\ref{tab-simi}. It is clear that the RBO values on all these datasets are beyond 0.5, which verifies our hypothesis. The RBO value between Train-Full is always higher than the one between Train-Test because most interactions are in the training set.

\begin{table}[t]
\footnotesize
\setlength{\extrarowheight}{.095pt}
\setlength\tabcolsep{3pt}
\caption{Efficiency comparison of CAGCN* with LightGCN. For fair comparison, we track the first time CAGCN* achieves the best performance of LightGCN.}
\label{tab-efficiency2}
\vskip -3ex
\centering
\begin{tabular}{llllllll}
 \Xhline{2\arrayrulewidth}
\textbf{Model} & \textbf{Stage} & \textbf{Gowalla} & \textbf{Yelp} & \textbf{Amazon} & \textbf{Ml-1M} & \textbf{Loseit} & \textbf{News} \\
 \Xhline{2\arrayrulewidth}
 LightGCN& Training & 16432.0 & 28788.0 & 81976.5 & 18872.3 & 39031.0 & 13860.8 \\

\Xhline{2\arrayrulewidth}
\multirow{3}{*}{CAGCN*} & Preprocess & 167.4 & 281.6 & 1035.8 & 33.8 & 31.4 & 169.0 \\
 & Training & 2963.2 & 1904.4 & 1983.9 & 11304.7 & 10417.7 & 1088.4 \\
 & Total & 3130.6 & 2186.0 & 3019.7 & 11338.5 & 10449.1 & 1157.4 \\
\Xhline{2\arrayrulewidth}
\multirow{2}{*}{\textbf{Improve}}& Training & 82.0\% & 93.4\% & 97.6\% & 40.1\% & 73.3\% & 92.1\%\\
 & Total & 80.9\% & 92.4\% & 96.3\% & 39.9\% & 73.2\% & 91.6\%\\
\Xhline{2\arrayrulewidth}
\end{tabular}
\vskip -2ex
\end{table}

\begin{figure}[t]
\vskip -1ex
     \centering
     \includegraphics[width=0.48\textwidth]{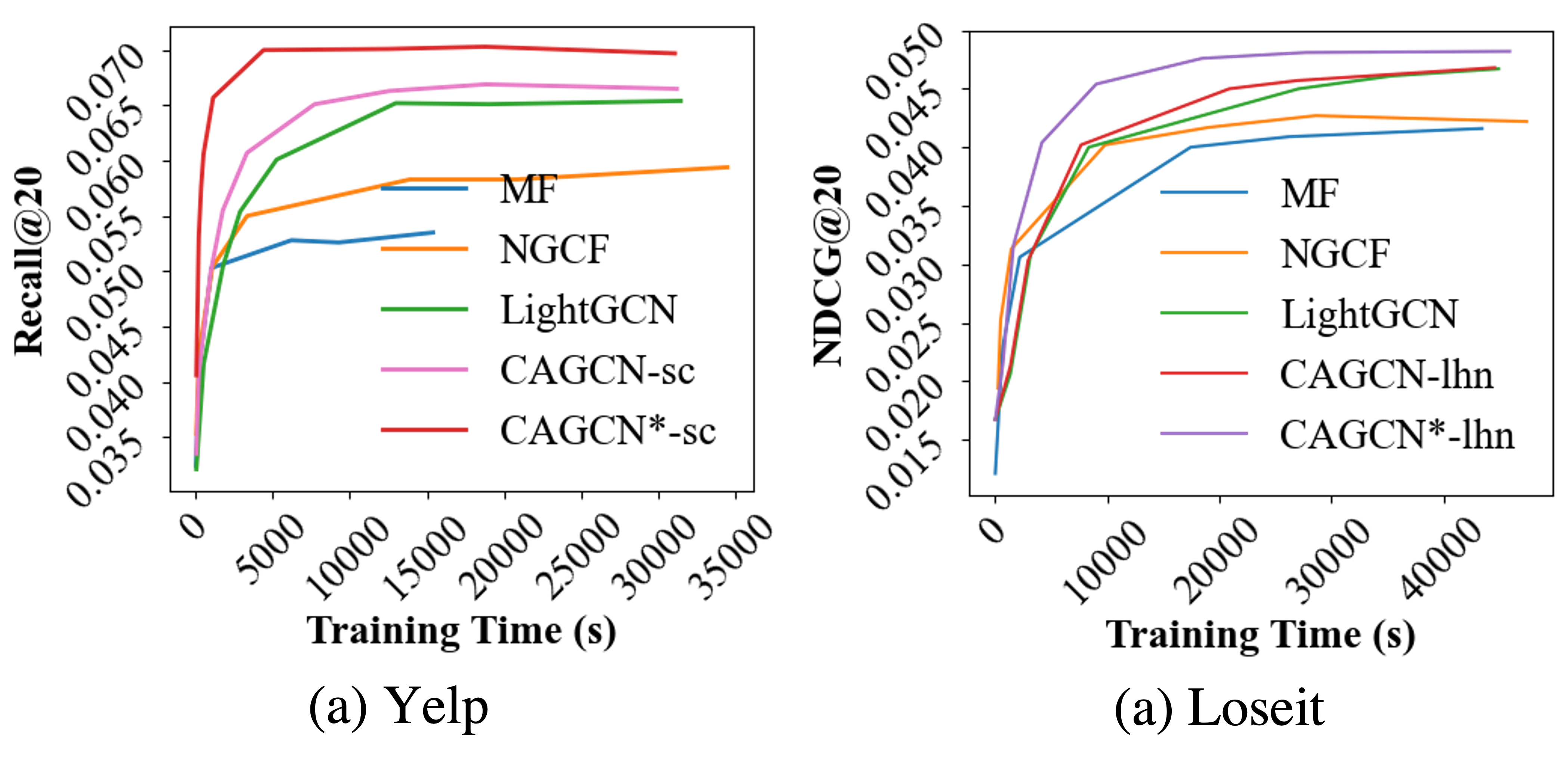}
     \vskip -2.25ex
     \vskip -1ex
     \caption{Training time (s) of different models.}
     \label{fig-train}
     \vskip -2ex
\end{figure}

Moreover, by combining two views of propagations, one from CAGC and one from LightGCN, CAGCN* achieves even higher performance with even less time. This is because keeping aggregating more information from neighbors with higher CIR (as CAGCN does) would prevent each user from aggregating information from his/her other neighbors. In addition, we report the first time that our best CAGCN* variant achieves the best performance of LightGCN on each dataset in Table~\ref{tab-efficiency2}. We also report the preprocessing time for pre-calculating the CIR matrix $\boldsymbol{\Phi}$ for our model to avoid any bias. We could see that even considering the preprocessing time, it still takes significantly less time for CAGCN* to achieve the same best performance as LightGCN, which highlights the broad prospects to deploy CAGCN* in real-world recommendations.

\begin{table}[t]
\scriptsize
\setlength{\extrarowheight}{.095pt}
\setlength\tabcolsep{3pt}
\caption{Average Rank-Biased Overlap (RBO) of the ranked neighbor lists between training (i.e., $\mathcal{N}_u^1)$ and testing/full (i.e., $\widehat{\mathcal{N}}_u^1$ and $\mathcal{N}_u^1\cup$ $\widehat{\mathcal{N}}_u^1$, respectively) dataset over all nodes $u \in \mathcal{U}$.}
\centering
\label{tab-simi}
\vskip -4ex
\begin{tabular}{l|cc|cc|cc}
\hline
 \multirow{2}{*}{Metric} & \multicolumn{2}{c|}{Gowalla} & \multicolumn{2}{c|}{Yelp} & \multicolumn{2}{c}{Ml-1M} \\
 & Train-Test & Train-Full & Train-Test & Train-Full & Train-Test & Train-Full \\
 \hline
JC & 0.604$\pm$0.129 & 0.902$\pm$0.084 & 0.636$\pm$0.124 & 0.897$\pm$0.081 & 0.848$\pm$0.092 & 0.978$\pm$0.019\\
SC & 0.611$\pm$0.127 & 0.896$\pm$0.084 & 0.657$\pm$0.124 & 0.900$\pm$0.077 & 0.876$\pm$0.077 & 0.983$\pm$0.015\\
LHN & 0.598$\pm$0.121 & 0.974$\pm$0.036 & 0.578$\pm$0.100 & 0.976$\pm$0.029 & 0.845$\pm$0.082 & 0.987$\pm$0.009\\
CN & 0.784$\pm$0.120 & 0.979$\pm$0.029 & 0.836$\pm$0.100 & 0.983$\pm$0.023 & 0.957$\pm$0.039 & 0.995$\pm$0.006 \\
\hline
\end{tabular}
\vskip -4ex
\end{table}

\subsection{Further Probe}
\subsubsection{Performance grouped by node degrees.} Here we group nodes by degree and visualize the average performance of each group. Comparing non-graph-based models (e.g., MF), graph-based models (e.g., LightGCN, CAGCN(*)) achieve higher performance for lower degree nodes $[0, 300)$ while lower performance for higher degree nodes $[300, \text{Inf})$. Since node degree follows the power-law distribution~\cite{stephen2009explaining}, the average performance of graph-based models is still higher than MF. On one hand, graph-based models leverage neighborhood to augment the weak supervision for low-degree nodes. On the other hand, they introduce noisy interactions for higher-degree nodes. It is also interesting to see the opposite performance trends under different evaluation metrics: NDCG prefers high-degree nodes while recall prefers low-degree nodes. This indicates that different evaluation metrics have different sensitivity to node degrees and an unbiased node-centric evaluator is desired.

\begin{figure}[htbp!]
     \centering
     \vskip -1ex
     \hspace{-2.1ex}
     \includegraphics[width=0.48\textwidth]{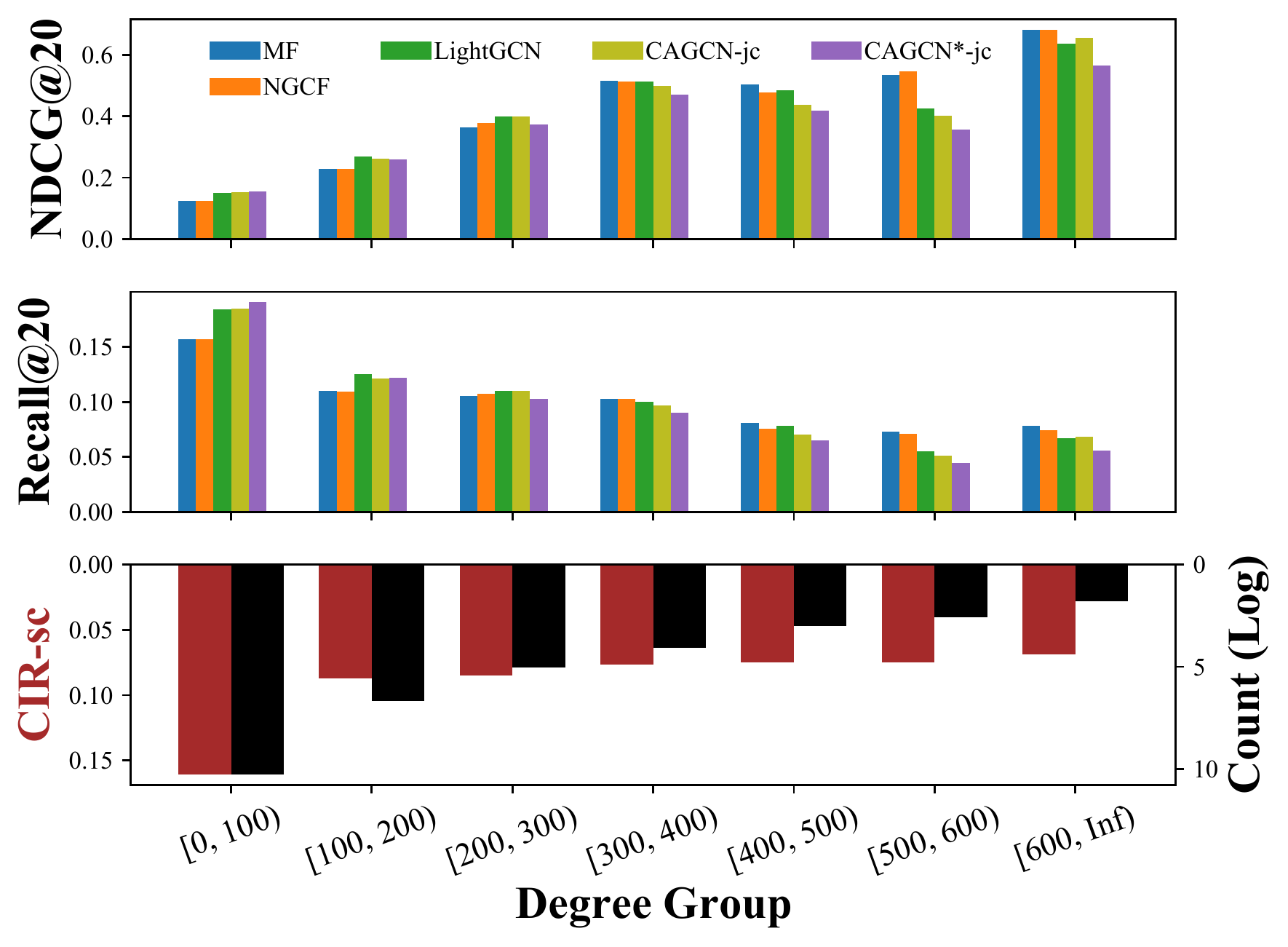}
     \vskip -2.25ex
     \caption{Performance w.r.t. node degree on Gowalla. A similar trend is seen on Yelp in {\textcolor{blue}{Supplementary}}~\ref{app-performinter}.}
     \label{fig-imbdeg}
     \vspace{-2.5ex}
\end{figure}

\begin{figure*}[htbp!]
    \centering
    \includegraphics[width=1\textwidth]{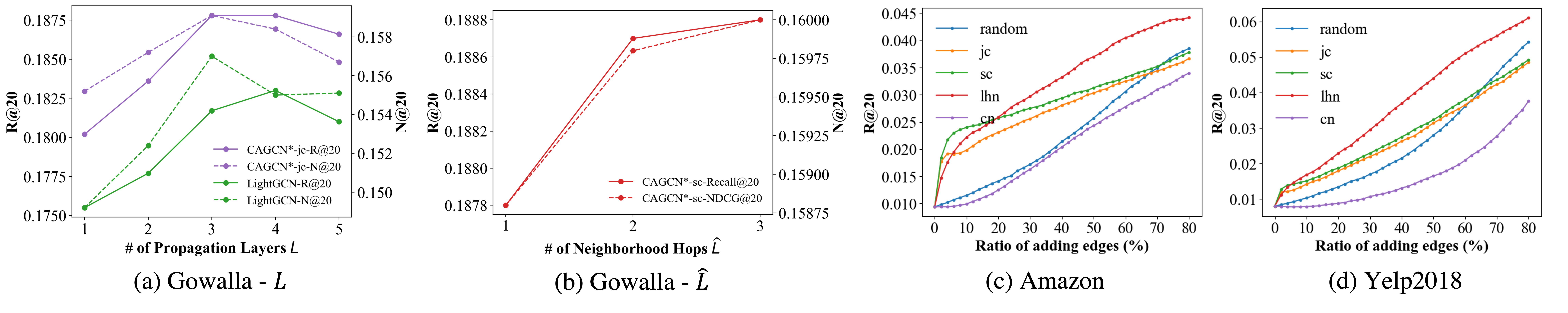}
    \vspace{-7ex}
    \caption{In (a)-(b), the performance first increases since we capture higher-layer neighborhood information and higher-hop topological interaction in calculating CIR as $L, \widehat{L}$ increase from 1 to 3. However, the performance decreases in (a) as $L$ increases due to over-smoothing. In (c)-(d), we add the global top edges directly (rather than cycle each node) according to their CIR. More details are provided in Appendix~\ref{app-res-construct2}.}
     \label{fig-analysis-empirically}
     \vskip -4ex
\end{figure*}

\vspace{-2ex}
\subsubsection{Impacts of propagation layers $L$ and neighborhood hops $\widehat{L}$.} Figure~\ref{fig-analysis-empirically}(a)-(b) visualize the performance of CAGCN* and LightGCN when the propagation layer $L$ in Eq.~\eqref{eq-mp} and the neighborhood hop $\widehat{L}$ in Eq.~\eqref{eq-phi-test} increase. In (a), the performance first increases as $L$ increases from 1 to 3 due to the incorporation of high-layer neighborhood information and then decreases due to over-smoothing. More importantly, our CAGCN* is always better than LightGCN at all propagation layers. In (b), the performance consistently increases as the number of neighborhood hops increases because we are allowed to consider even more higher topological interactions among each node's neighborhood in computing CIR.

\vspace{-1.75ex}
\subsubsection{Adding edges globally according to CIR.} Figure~\ref{fig-analysis-empirically}(c)-(d) visualize the performance change when we add edges randomly and according to CIR. Unlike Figure~\ref{fig-cir-retrain}-\ref{fig-cir-pretrain} where we add edges by cycling each node, here we directly select the global top edges regardless of each center node according to their CIR and then evaluate the LightGCN with the pre-trained user-item embeddings. In the first stage, we observe a similar trend that adding edges according to JC, SC, and LHN leads to faster performance gain. However, since we don't cycle over each node, we would keep adding so many edges with larger CIR to the same node, which fails to bring performance gain anymore and hence cannot maximize our performance benefit under the node-centric evaluation metric.

\vspace{-2ex}
\section{Related Work}\label{sec-relatedwork}
\textbf{Collaborative Filtering \& Recommendation.} Collaborative filtering (CF) predicts users' interests by utilizing the preferences of other users with similar interests~\cite{goldberg1992using}. Early CF methods used Matrix Factorization techniques~\cite{rendle2009bpr, koren2009matrix, bpr, tay2018latent} to capture CF effect via optimizing users/items' embeddings over historical interactions. Stepping further, Graph-based methods either leverage topological constraints or message-passing to inject the CF effect into user/item embeddings~\cite{lightgcn, ngcf}. ItemRank and BiRank~\cite{gori2007itemrank, he2016birank} perform label propagation and compute users' ranking based on structural proximity between the observed and the target items. To make user preferences learnable, HOP-Rec~\cite{yang2018hop} combines the graph-based method and the embedding-based method. Yet, interactions captured by random walks there do not fully explore the high-layer neighbors and multi-hop dependencies~\cite{tdgnn}. By contrast, GNN-based methods are superior at encoding higher-order structural proximity in user/item embeddings~\cite{ngcf, lightgcn}. Recent work~\cite{chen2021structured, GTN, tian2022learning} has demonstrated that not all captured collaborations improve users' ranking. \cite{chen2021structured} proposes to learn binary mask and impose low-rank regularization while ours propose novel topological metric CIR to weigh neighbors’ importance. \cite{GTN} smooths nodes’ embeddings based on degree-normalized embedding similarity, while ours adaptively smooth based on topological proximity(CIR).  \cite{tian2022learning} denoises interactions/preserve diversity based on 1-layer propagated embeddings and hence cannot go beyond 1-WL test, while ours keep neighbors and does not focus on diversity issues.

\noindent\textbf{Link Prediction.} As a generalized version of recommendation, link prediction finds applications in predicting drug interactions and completing knowledge graphs~\cite{rozemberczki2022chemicalx, nickel2015review}. Early studies adopt topological heuristics to score node pairs~\cite{leicht2006vertex, zhou2009predicting, newman2001clustering}. Furthermore, latent-based/deep-learning methods~\cite{perozzi2014deepwalk, zhang2017weisfeiler} are proposed to characterize underline topological patterns in node embeddings via random walks~\cite{grover2016node2vec} or regularizing~\cite{perozzi2014deepwalk}. To fully leverage node features, GNN-based methods are proposed and achieve unprecedented success owing to the use of the neural network to extract task-related information and the message-passing capture the topological pattern~\cite{seal, walkpool, counterfactual}. Recently, efforts have been invested in developing expressive GNNs that can go beyond the 1-WL test \cite{zhao2021stars, beyond, liu2022your} for node/graph classification. Following this line, our work develops a recommendation-tailored graph convolution with provably expressive power in predicting links between users and items.

\vspace{-2ex}
\section{Conclusion}\label{sec-conclusion}
In this paper, we find that the message-passing captures collaborative effect by leveraging interactions between neighborhoods. The strength of the captured collaborative effect depends the embedding similarity, the weight of paths and the contribution of each propagation layer. To determine whether the captured collaborative effect would benefit the prediction of user preferences, we propose the Common Interacted Ratio (CIR) and empirically verify that leveraging collaborations from neighbors with higher CIR contributes more to users' ranking. Furthermore, we propose CAGCN(*) to selectively aggregate neighboring nodes' information based on their CIRs. We further define a new type of isomorphism, bipartite-subgraph-isomorphism, and prove that our CAGCN* can be more expressive than 1-WL in distinguishing subtree(subgraph)-isomorphic yet non-bipartite-subgraph-isomorphic graphs. Experimental results demonstrate the advantages of the proposed CAGCN(*) over other baselines. Specifically, CAGCN* outperforms the most representative graph-based recommendation model, LightGCN~\cite{lightgcn}, by around 10\% in Recall@20 but also achieves roughly more than 80\% speedup. 
In the future, we will explore the imbalanced performance improvement among nodes in different degree groups as seen in Figure~\ref{fig-imbdeg}, especially from the perspective of GNN fairness~\cite{wang2022wsdm,wang2022improving}. 

\bibliographystyle{ACM-Reference-Format}
\bibliography{references}

\appendix
\newpage 
\section{Appendix}
\subsection{Graph Topological Metrics for CIR}\label{sec-toponote}
Here we demonstrate that by configuring different $f$ and $\widehat{L}$, $\phi_{u}^{\widehat{L}}(j)$ can express many existing graph similarity metrics.

\vspace{-3ex}
\begin{equation}
    \phi^{\widehat{L}}_u(j) = \frac{1}{|\mathcal{N}_u^1|}\sum_{i\in\mathcal{N}_u^1}\sum_{l = 1}^{\widehat{L}}\beta^{2l}\sum_{P_{ji}^{2l}\in\mathscr{P}_{ji}^{2l}}{\frac{1}{f(\{\mathcal{N}_k^1|k\in P_{ji}^{2l}\})}}
\end{equation}

\begin{itemize}[leftmargin=*]
\item \textbf{Jaccard Similarity (JC)~\cite{liben2007link}:} The JC score measures the similarity between neighborhood sets as the ratio of the intersection of two neighborhood sets to the union of these two sets: 
\vspace{-1ex}
\begin{equation}
\small
    \text{JC}(i, j)=\frac{|\mathcal{N}^1_i\cap\mathcal{N}^1_j|}{|\mathcal{N}^1_i\cup\mathcal{N}^1_j|}
\end{equation}

Let $\widehat{L} = 1$ and set $f(\{\mathcal{N}_k^1|k\in P_{ji}^{2}\}) = |\mathcal{N}_i^1 \cup \mathcal{N}_j^1|$, then we have:
\begin{equation}
\tiny
    \phi^1_u(j) = \frac{1}{|\mathcal{N}_u^1|}\sum_{i\in\mathcal{N}_u^1}\beta^{2}\sum_{P_{ji}^{2}\in\mathscr{P}_{ji}^{2}}\frac{1}{|\mathcal{N}_i^1 \cup \mathcal{N}_j^1|} = \frac{\beta^{2}}{|\mathcal{N}_u^1|}\sum_{i\in\mathcal{N}_u^1}\frac{|\mathcal{N}_i^1 \cap \mathcal{N}_j^1|}{|\mathcal{N}_i^1 \cup \mathcal{N}_j^1|} = \frac{\beta^{2}}{|\mathcal{N}_u^1|}\sum_{i\in\mathcal{N}_u^1}\text{JC}(i, j)
\end{equation}

\item \textbf{Salton Cosine Similarity (SC)~\cite{salton1989automatic}:} The SC score measures the cosine similarity between the neighborhood sets of two nodes: 
\begin{equation}
\small
    \text{SC}(i, j)=\frac{|\mathcal{N}^1_i\cap\mathcal{N}^1_j|}{\sqrt{|\mathcal{N}^1_i\cup\mathcal{N}^1_j|}}
\end{equation}

let $\widehat{L} = 1$ and set $f(\{\mathcal{N}_k^1|k\in P_{ji}^{2}\}) = \sqrt{|\mathcal{N}_i^1 \cup \mathcal{N}_j^1|}$, then we have:
\begin{equation}
\tiny
    \phi^1_u(j) = \frac{1}{|\mathcal{N}_u^1|}\sum_{i\in\mathcal{N}_u^1}\beta^{2}\sum_{P_{ji}^{2}\in\mathscr{P}_{ji}^{2}}{\frac{1}{\sqrt{|\mathcal{N}_i^1 \cup \mathcal{N}_j^1|}}} = \frac{\beta^{2}}{|\mathcal{N}_u^1|}\sum_{i \in \mathcal{N}_u^1}\frac{|\mathcal{N}_i^1 \cap \mathcal{N}_j^1|}{\sqrt{|\mathcal{N}_i^1 \cup \mathcal{N}_j^1|}} = \frac{\beta^{2}}{|\mathcal{N}_u^1|}\sum_{i\in\mathcal{N}_u^1}\text{SC}(i, j)
\end{equation}
\vspace{-3ex}
\item \textbf{Common Neighbors (CN)~\cite{newman2001clustering}:} The CN score measures the number of common neighbors of two nodes and is frequently used for measuring the proximity between two nodes: 
\begin{equation}
\small
    \text{CN}(i, j)= |\mathcal{N}^1_i\cap\mathcal{N}^1_j|
\end{equation}

Let $\widehat{L} = 1$ and set $f(\{\mathcal{N}_k^1|k\in P_{ji}^{2}\}) = 1$, then we have:
\begin{equation}
\tiny
    \phi^1_u(j) = \frac{1}{|\mathcal{N}_u^1|}\sum_{i\in\mathcal{N}_u^1}\beta^{2}\sum_{P_{ji}^{2}\in\mathscr{P}_{ji}^{2}}1 = \frac{\beta^{2}}{|\mathcal{N}_u^1|}\sum_{i \in \mathcal{N}_u^1}|\mathcal{N}_i^1 \cap \mathcal{N}_j^1| = \frac{\beta^{2}}{|\mathcal{N}_u^1|}\sum_{i\in\mathcal{N}_u^1}\text{CN}(i, j)
\end{equation}

Since CN does not contain any normalization to remove the bias of degree in quantifying proximity and hence performs worse than other metrics as demonstrated by our recommendation experiments in Table~\ref{tab-resfull}.

\item \textbf{Leicht-Holme-Nerman (LHN)~\cite{leicht2006vertex}:} LHN is very similar to SC. However, it removes the square root in the denominator and is more sensitive to the degree of node: 
\vspace{-1ex}
\begin{equation}
\small
    \text{LHN}(i, j) = \frac{|\mathcal{N}_i^1\cap\mathcal{N}_j^1|}{|\mathcal{N}_i^1|\cdot|\mathcal{N}_j^1|}
\end{equation}

Let $\widehat{L} = 1$ and set $f(\{\mathcal{N}_k^1|k\in P_{ji}^{2}\}) = |\mathcal{N}_i^1|\cdot|\mathcal{N}_j^1|$, then we have:
\begin{equation}
\tiny
    \phi^1_u(j) = \frac{1}{|\mathcal{N}_u^1|}\sum_{i\in \mathcal{N}_u^1}\beta^{2}\sum_{P_{ji}^{2}\in\mathscr{P}_{ji}^{2}}\frac{1}{|\mathcal{N}_i^1|\cdot|\mathcal{N}_j^1|} = \frac{\beta^{2}}{|\mathcal{N}_u^1|}\sum_{i\in\mathcal{N}_u^1}\frac{|\mathcal{N}_i^1 \cap \mathcal{N}_j^1|}{|\mathcal{N}_i^1|\cdot|\mathcal{N}_j^1|} = \frac{\beta^{2}}{|\mathcal{N}_u^1|}\sum_{i\in\mathcal{N}_u^1}\text{LHN}(i, j)
\end{equation}

We further emphasize that our proposed CIR is a generalized version of these four existing metrics and can be delicately designed toward satisfying downstream tasks and datasets. We leave such exploration on the choice of $f$ as one potential future work.
\end{itemize}

\subsection{Derivation of Eq.~\eqref{eq-collaborative}}\label{sec-proof1}
The matrix form of computing the ranking of user $u$ over item $i$ after $L$-layer LightGCN-based message-passing:

\begin{equation}\label{eq-dot}
    y_{ui}^{L} = (\sum_{l_1= 0}^{L}{\beta_{l_1}\mathbf{E}_u^{l_1}})^{\top}(\sum_{l_1= 0}^{L}{\beta_{l_1}\mathbf{E}_i^{l_1}}) = (\sum_{l_1 = 0}^L\beta_{l_1}\mathbf{A}^{l_1}\mathbf{E}^0)_{u}^{\top}(\sum_{l_1 = 0}^L\beta_{l_1}\mathbf{A}^{l_1}\mathbf{E}^0)_{i}.
\end{equation}
where $\beta_{l_1}$ is the layer contribution and LightGCN uses mean-pooling, i.e., $\frac{1}{L}$ in Eq.~\eqref{eq-mp}. For the propagated embedding at a specific layer $l_1$, we have:
\begin{equation}\label{eq-prop1}
    \mathbf{E}^{l_1}_u = (\mathbf{A}^{l_1}\mathbf{E}^0)_u = \sum\limits_{j\in\mathcal{V}_u^{l_1}}\alpha_{ju}^{l_1}\mathbf{e}_j^0,
\end{equation}
where $\alpha_{ju}^{l_1} = \sum_{P_{ju}^{l_1}\in\mathscr{P}_{ju}^{l_1}}\prod_{e_{pq}\in P_{ju}^{l_1}}{d^{-0.5}_{p}d^{-0.5}_{q}}$($\alpha_{ju}^{l_1} = 0$ if $\mathscr{P}_{ju}^{l_1}=\emptyset$). $\mathcal{V}_u^{l_1}$ is the set of all nodes having paths of length $l_1$ to $u$ and can be expressed as:
\begin{equation}\label{eq-prop2}
    \mathcal{V}_u^{l_1} = \bigcup_{l_2 = 0}^{l_1}\mathcal{N}_u^{l_2}\cdot\mathds{1}[(l_1 - l_2)\%2 = 0],
\end{equation}
where
\begin{equation}\label{eq-prop3}
    \mathcal{N}^{l_2}_u\cdot\mathds{1}[(l_1 - l_2)\%2 = 0] = \begin{cases}
    \mathcal{N}^{l_2}_u, & (l_1 - l_2)\%2 = 0\\
    \emptyset, & (l_1 - l_2)\%2 \ne 0
    \end{cases}.
\end{equation}

Substituting Eq.~\eqref{eq-prop2} into Eq.~\eqref{eq-prop1}, we have:
\begin{equation}\label{eq-prop4}
\tiny
    \mathbf{E}^{l_1}_u = (\mathbf{A}^{l_1}\mathbf{E}^0)_u = \hspace{-1.5ex}\sum\limits_{j\in\mathcal{V}_u^{l_1}}\alpha_{ju}^{l_1}\mathbf{e}_j^0 = \hspace{-3.5ex}\sum\limits_{j\in\bigcup_{l_2 = 0}^{l_1}\mathcal{N}^{l_2}_u\cdot\mathds{1}[(l_1 - l_2)\%2 = 0]} \hspace{-3.5ex} \alpha_{ju}^{l_1}  \mathbf{e}_j^0 = \sum_{l_2 = 0}^{l_1}\sum_{j \in \mathcal{N}_u^{l_2} \cdot \mathds{1}[(l_1 - l_2)\%2 = 0]} \hspace{-3.5ex}{\alpha_{ju}^{l_1}\mathbf{e}_j^0}.
\end{equation}
Then the aggregation of all $L$ layers' embeddings of user $u$ is expressed as:
\begin{equation}\label{eq-prop5}
    \sum_{l_1= 0}^{L}{\beta_{l_1}\mathbf{E}_u^{l_1}} = \sum_{l_1 = 0}^{L}\beta_{l_1}\sum_{l_2 = 0}^{l_1}\sum_{j \in \mathcal{N}_u^{l_2} \cdot \mathds{1}[(l_1 - l_2)\%2]}{\alpha_{ju}^{l_1}\mathbf{e}_j^0}.
\end{equation}
Eq.~\eqref{eq-prop5} means that for each length $l_1\in \{0, 1, ..., L\}$, for each node $j\in\mathcal{V}_u^{l_1}$ that has path of length $l_1$ to $u$, we propagate its embedding over each path $P_{ju}^{l_1}\in \mathscr{P}_{ju}^{l_1}$ with the corresponding weight coefficient $\prod_{e_{pq}\in P_{ju}^{l_1}}{d^{-0.5}_{p}d^{-0.5}_{q}}$.

Since nodes that are $l_1$-hops away from $u$ cannot have paths of length less than $l_1$, we reorganize Eq.~\eqref{eq-prop5} by first considering the hop of each node and then considering the length of each path, which leads to:
\begin{equation}\label{eq-prop6}
\footnotesize
    \sum_{l_1= 0}^{L}{\beta_{l_1}\mathbf{E}_u^{l_1}} = \sum_{l_1 = 0}^{L}\beta_{l_1}\sum_{l_2 = 0}^{l_1}\sum_{j \in \mathcal{N}^{l_2}_u \cdot \mathds{1}[(l_1 - l_2)\%2]}{\alpha_{ju}^{l_2}\mathbf{e}_j^0} = \sum_{l_1 = 0}^{L}\sum_{j\in\mathcal{N}_u^{l_1}}\sum_{l_2 = l_1}^{L}\beta_{l_2}\alpha_{ju}^{l_2}\mathbf{e}_j^0,
\end{equation}
where $\alpha_{ju}^{l_2} = \sum_{P_{ju}^{l_2}\in\mathscr{P}_{ju}^{l_2}}\prod_{e_{pq}\in P_{ju}^{l_2}}{d^{-0.5}_{p}d^{-0.5}_{q}}$($\alpha_{ju}^{l_2} = 0$ if $\mathscr{P}_{ju}^{l_2}=\emptyset$). Then by substituting Eq.~\eqref{eq-prop6} into Eq.~\eqref{eq-dot}, we end up with:

\begin{equation}\label{eq-prop7}
y_{ui}^L = (\sum_{l_1=0}^{L}\sum_{j\in\mathcal{N}^{l_1}_{u}}\sum_{l_2 = l_1}^{L}\beta_{l_2}\alpha_{ju}^{l_2}\mathbf{e}_{j}^0)^{\top}(\sum_{l_1=0}^{L}\sum_{v\in\mathcal{N}^{l_1}_{i}}\sum_{l_2 = l_1}^{L}\beta_{l_2}\alpha_{vi}^{l_2}\mathbf{e}_{v}^0),
\end{equation}
\noindent where $\mathcal{N}^{0}_u = \{u\}$ and specifically, $\alpha_{uu}^{0} = 1$. $\beta_{l_2}$ is the weight measuring contributions of propagated embeddings at layer $l_2$.

\subsection{Complexity Comparison and Analysis}\label{sec-complexity}
Let $|\mathcal{V}|, |\mathcal{E}|, |\mathcal{F}|$ be the total number of nodes, edges, and feature dimensions (assuming feature dimensions stay the same across all feature transformation layers). Let $L$ be the propagation layer for all graph-based models using message-passing. Let $r$ be the total number of negative samples per epoch per positive pair and $K$ be the number of $2^{\text{nd}}$-order neighbors. For $r$, all baselines use 1 per epoch per positive pair and hence can be omitted (aside from UltraGCN using a larger number). Then the complexity of each model is summarized in Table~\ref{tab-complexity}. For CAGCN, since we only consider $2$-hops away connections to compute CIR in Eq.~\eqref{eq-phi-test}, the main computational load would be computing the power of adjacency matrix, which takes $\mathcal{O}(|\mathcal{V}|^3)$. 
Note that for both of our CAGCN and UltraGCN, we can apply Strassens's Algorithm to further reduce the $\mathcal{O}(|\mathcal{V}|^3)$ to $\mathcal{O}(|\mathcal{V}|^{2.8})$. In Table~\ref{tab-efficiency2} in Section~\ref{sec-effcompare}, we report the preprocessing time for each dataset. Clearly, compared with the time used for training, the time for preprocessing is minor, which even demonstrates the superior efficiency of CAGCN since it significantly speeds up the training as justified in Section~\ref{sec-effcompare}.

\vspace{-3ex}
\begin{table}[htbp!]
\centering
\scriptsize 
\setlength{\extrarowheight}{0pt}
\setlength\tabcolsep{2pt}
\caption{Complexity of the pre-procession and the forward pass of CAGCN and different baselines.}
\vspace{-3ex}
\label{tab-complexity}
\begin{tabular}{llccc}
\hline
\multicolumn{2}{c}{Model} & MF & NGCF & LightGCN \\
\hline
\multicolumn{2}{l}{ 
\begin{tabular}{@{}c@{}} \# Extra Hyper-parameters \end{tabular}
} & / & / & 1 \\
\hline
\multirow{2}{*}{Preprocess} & Space & / & $\mathcal{O}(|\mathcal{E}| + |\mathcal{V}|)$ & $\mathcal{O}(|\mathcal{E}| + |\mathcal{V}|)$ \\
 & Time & / & $\mathcal{O}(|\mathcal{E}| + |\mathcal{V}|)$ & $\mathcal{O}(|\mathcal{E}| + |\mathcal{V}|)$\\
 \hline
\multirow{2}{*}{Training} & Space & $\mathcal{O}(|\mathcal{V}|F)$ & $\mathcal{O}(L|\mathcal{V}|F + |\mathcal{E}| + LF^2)$ & $\mathcal{O}(L|\mathcal{V}|F + |\mathcal{E}|)$ \\
 & Time & $\mathcal{O}(|\mathcal{E}|F)$ & $\mathcal{O}(L(|\mathcal{E}|F + |\mathcal{V}|F^2))$ & $\mathcal{O}(L|\mathcal{E}|F + L|\mathcal{V}|F)$\\
 \hline
\end{tabular}

\begin{tabular}{llccc}
\hline
\multicolumn{2}{c}{Model} & GTN & UltraGCN & CAGCN \\
\hline
\multicolumn{2}{l}{ 
\begin{tabular}{@{}c@{}} \# Extra Hyper-parameters \end{tabular}
} & 1 & 7 & 2 \\
\hline
\multirow{2}{*}{Preprocess} & Space & $\mathcal{O}(|\mathcal{E}| + |\mathcal{V}|)$ & $\mathcal{O}(|\mathcal{E}| + |\mathcal{V}|)$ & $\mathcal{O}(|\mathcal{E}| + |\mathcal{V}|)$ \\
 & Time & $\mathcal{O}(|\mathcal{E}| + |\mathcal{V}|)$ & $\mathcal{O}(|\mathcal{V}|^3)$ & $\mathcal{O}(|\mathcal{V}|^3)$ \\
 \hline
\multirow{2}{*}{Training} & Space  & $\mathcal{O}(L|\mathcal{V}|F + |\mathcal{E}|)$ & $\mathcal{O}(|\mathcal{V}|F + |\mathcal{V}|K)$ & $\mathcal{O}(L|\mathcal{V}|F + |\mathcal{E}|)$ \\
 & Time  & $\mathcal{O}(L|\mathcal{E}|F + L|\mathcal{V}|F)$  & $\mathcal{O}(r(|\mathcal{E}| + |V|K)F)$ & $\mathcal{O}(L|\mathcal{E}|F + L|\mathcal{V}|F)$\\
 \hline
\end{tabular}
\vspace{-2ex}
\end{table}

\vspace{-4ex}
\subsection{Experimental Setting}
\subsubsection{Baselines}\label{app-baseline}
We compare our proposed CAGCN(*) with the following baselines: \textbf{MF~\cite{bpr}:} Most classic collaborative filtering method equipped with the BPR loss; \textbf{NGCF~\cite{ngcf}:} The first GNN-based collaborative filtering model; 
\textbf{LightGCN~\cite{lightgcn}:} The most popular GNN-based collaborative filtering model, which removes feature transformation and nonlinear activation; \textbf{UltraGCN~\cite{ultragcn}:} The first model approximating regularization weights by infinite layers of message passing, and leveraging higher-order user-user relationships; 
\textbf{GTN~\cite{GTN}:} This model leverages a robust and adaptive propagation based on the trend of the aggregated messages to avoid unreliable user-item interactions.

\vspace{-1ex}
\subsubsection{CAGCN(*)-variants}\label{app-variant}
For CAGCN, 
$\gamma_i = \sum_{r\in\mathcal{N}_i^1}d_i^{-0.5}d_r^{-0.5}$ to ensure that the total edge weights for messages received by each node are the same as LightGCN. Therefore, Eq.~\eqref{eq-CAGCNaug} becomes:
\vspace{-1ex}
\begin{equation}\label{eq-CAGCNaug1}
    \mathbf{e}_i^{l+1} = \sum_{j\in\mathcal{N}_i^1}((\sum_{r\in\mathcal{N}_i^1}d_i^{-0.5}d_r^{-0.5})\frac{\boldsymbol{\Phi}_{ij}}{\sum_{k\in\mathcal{N}_i^1}{\boldsymbol{\Phi}_{ik}}})\mathbf{e}_j^l, \forall i\in\mathcal{V}.
\end{equation}

\noindent For CAGCN*, 
$\gamma_i = \gamma$ as a constant controlling the trade-off between contributions from message-passing according to LightGCN and according to CAGC. Eq.~\eqref{eq-CAGCNaug} becomes:
\vspace{-1ex}
\begin{equation}\label{eq-CAGCNaug2}
    \mathbf{e}_i^{l+1} = \sum_{j\in\mathcal{N}_i^1}(\gamma\frac{\boldsymbol{\Phi}_{ij}}{\sum_{k\in\mathcal{N}_i^1}{\boldsymbol{\Phi}_{ik}}} + d_{i}^{-0.5}d_{j}^{-0.5})\mathbf{e}_j^l, \forall i\in\mathcal{V},
\end{equation}
where we search $\gamma$ in $\{1, 1.2, 1.5, 1.7, 2.0\}$.

\vspace{-1ex}
\subsection{Additional Experiments}

\subsubsection{Adding edges according to local CIRs}\label{app-res-construct}
Given the user-item bipartite graph for training, we calculate the CIR-variants and use them to rank the neighborhood for each center node. During construction, we first remove all edges and then iteratively cycle over each node and add its corresponding neighbor based on the ranking until hitting the budget. Figure~\ref{fig-app_verify}(a) contains an example with 
users $u_1, u_2$ and a budget of three edges, where $u_1$ and $u_2$ both first get an edge, but then only $u_1$ gets a second edge. 

\begin{figure}[htbp!]
     \centering
     \vskip -2ex
     \hspace{-1ex}\includegraphics[width=0.48\textwidth]{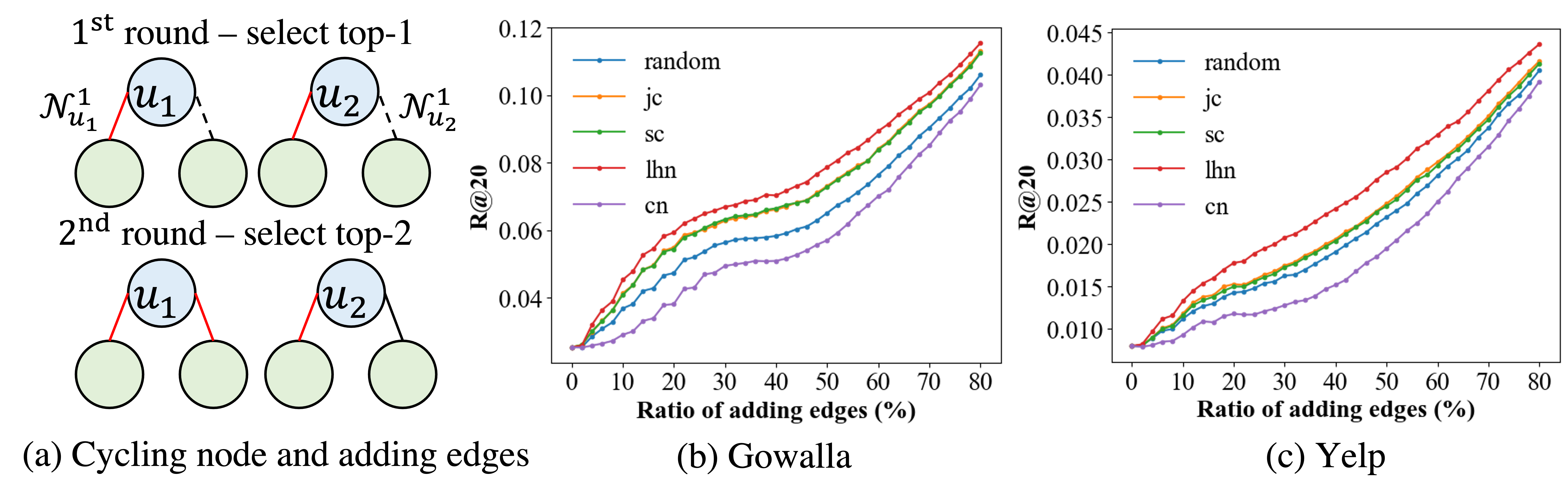}
     \vskip -3.5ex
     \caption{(a) The procedure of adding edges according to CIR of neighbors around each node. (b)-(c) The performance change of adding edges on Gowalla and Yelp.}
     \vskip -2.5ex
     \label{fig-app_verify}
\end{figure}

Similar to what we observed in Figure~\ref{fig-cir-pretrain}, the performance increases as we add more edges on Gowalla and Yelp (Figure~\ref{fig-app_verify}(b) and (c), respectively). Furthermore, except for cn, adding edges according to CIR-variants is more effective in increasing the performance, which demonstrates the effectiveness of CIR in measuring the edge importance.

\subsubsection{Adding edges according to global CIRs}\label{app-res-construct2}
Here we introduce how we add edges globally according to CIRs. Given the user-item interactions for training, we first construct the user-item bipartite graph and calculate the different variants of CIR including jc, sc, cn, lhn as stated in Appendix~\ref{sec-toponote}. Then, we directly rank all edges according to the computed CIR. In the construction stage, we first remove all edges in the bipartite graph. Then we select the top edges according to the ranking based on our budget. 
Figure~\ref{fig-app_verify2}(a) contains an example with users $u_1, u_2$ and a budget of three edges, where we directly select the top-3 edges from all users' neighbors.

\begin{figure}[htbp!]
     \centering
     \vskip -2ex
     \hspace{-1ex}\includegraphics[width=0.48\textwidth]{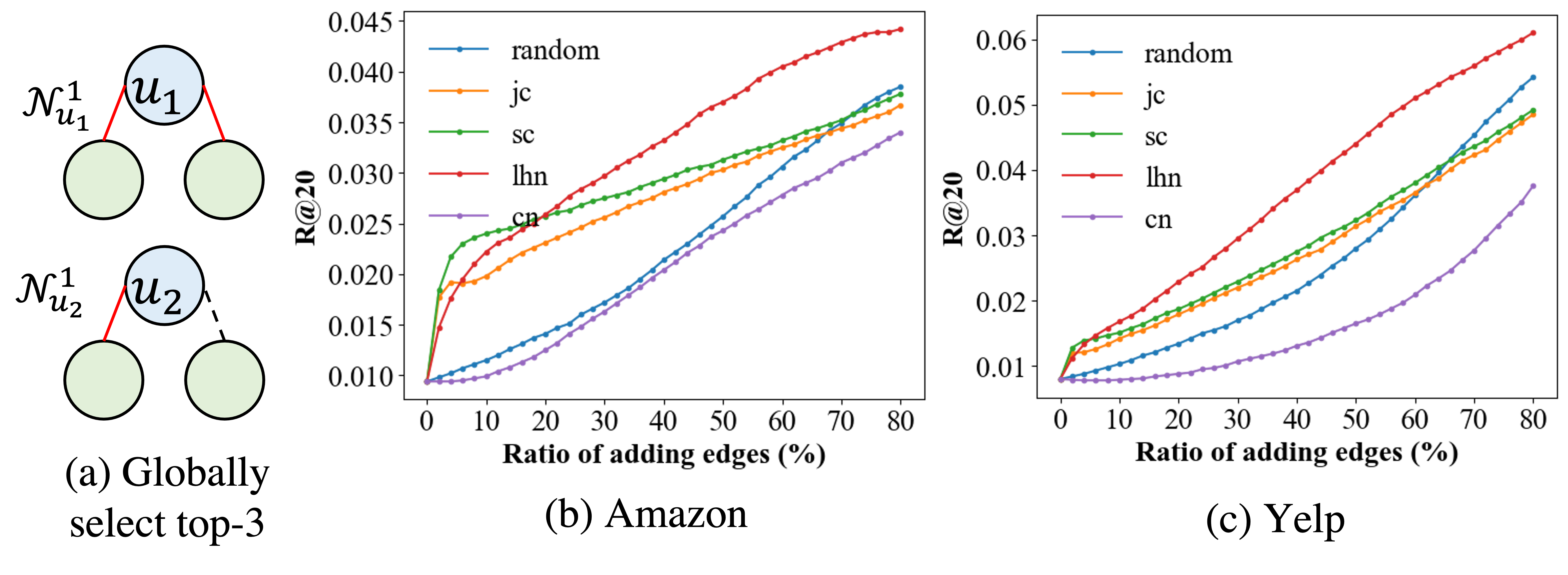}
     \vskip -3.5ex
     \caption{(a) The procedure of adding edges according to CIR globally. (b)-(c) The performance change of adding edges on Amazon and Yelp.}
    \vskip -2ex
     \label{fig-app_verify2}
\end{figure}

In the first stage, we observe a similar trend that adding edges according to CIRs lead to faster performance gain as Figure~\ref{fig-app_verify}, which demonstrate the effectiveness of CIR in measuring the edge importance globally. However, since we don't cycle over each node and add its corresponding edge as we do in Appendix~\ref{app-res-construct}, we would keep adding so many edges with larger CIR to the same node, which may not maximize our performance benefit when the metric is calculated by averaging over all nodes.

\newpage
\section{Supplementary}\label{sec-supp}

\subsection{Hyperparamters}
We follow the procedure of hyperparameter tuning in~\cite{lightgcn, ngcf} and list the hyperparameters as follows:
\vspace{-1ex}
\begin{itemize}[leftmargin=*]
    \item \textbf{LightGCN.} Propagation layers: $L = 3$; Pooling layer: Meaning pooling; 
    
    \item \textbf{NGCF.} Propagation layers: $L=3$; Slope of LeakyRelu: $0.2$; Pooling layer: Concatenation

    \item \textbf{UltraGCN.} For Gowalla, Yelp, Amazon and Ml-1M, we use exactly the same hyperparameter configurations provided \href{https://github.com/xue-pai/UltraGCN}{here}. For Loseit and News, the hyperparamters are as follows:
    \newline\indent (1) Loseit: Training epochs 2000; Learning rate $1e^{-3}$; batch size 512; Loss weights $w_1 = 1e^{-6}, w_2=1, w_3=1e^{-6}, w_4=1$; the number of negative samples per epoch per positive pair $20$; negative weight $20$; weight of $l_2$ regularization $\gamma=1e^{-4}$, $2^{\text{nd}}$-constraining loss coefficient $\lambda=5e^{-4}$.
    \newline\indent (2) News: Training epochs 2000; Learning rate $1e-3$; batch size 1024; Loss weights $w_1 = 1e^{-8}, w_2=1, w_3=1, w_4=1e^{-8}$; the number of negative samples per epoch per positive pair $1000$; negative weight $200$; weight of $l_2$ regularization $\gamma=1e^{-4}$, $2^{\text{nd}}$-constraining loss coefficient $\lambda=5e^{-4}$.

    \item \textbf{GTN.} For Gowalla, Yelp, Amazon, we directly report the result provided \href{https://github.com/wenqifan03/gtn-sigir2022}{here}.
    In the following, we introduce the hyparamemeters we used for our CAGCN(*)-variants. With specification, the number of training epochs is set to be 1000; the learning rate 0.001; $l_2$ regularization $1e^{-4}$; number of negative samples $1$; embedding dimenstion $64$; batch size $256$; $\widehat{L} = 1$.
    \item \textbf{CAGCN-jc.} (1) Gowalla: $\gamma=1$; (2) Yelp: $\gamma=1.2$; (3) Amazon: $\gamma=1$; (4) Ml-1M: $\gamma=2$; (5) Loseit: $\gamma=1$; (6) News: $\gamma=1$.

    \item \textbf{CAGCN-cn.} (1) Gowalla: $\gamma=1$; (2) Yelp: $\gamma=1.2$; (3) Amazon: $\gamma=1$; (4) Ml-1M: $\gamma=1$; (5) Loseit: $\gamma=1$; (6) News: $\gamma=1$.

    \item \textbf{CAGCN-sc.} (1) Gowalla: $\gamma=1$; (2) Yelp: $\gamma=1$; (3) Amazon: $\gamma=1$; (4) Ml-1M: $\gamma=2$; (5) Loseit: $\gamma=1$; (6) News: $\gamma=1$.

    \item \textbf{CAGCN-lhn.} (1) Gowalla: $\gamma=1.2$; (2) Yelp: $\gamma=1$; (3) Amazon: $\gamma=1$; (4) Ml-1M: $\gamma=2$; (5) Loseit: $\gamma=1, L=1$; (6) News: $\gamma=1.5$.
    
    \item \textbf{CAGCN*-jc.} (1) Gowalla: $\gamma=1.2$, $l_2$-regularization $1e-3$; (2) Yelp: $\gamma=1.7$, $l_2$-regularization $1e^{-3}$; (3) Amazon: $\gamma=1.7$, $l_2$-regularization $1e^{-3}$; (4) Ml-1M: $\gamma=1$, $l_2$-regularization $1e^{-3}$; (5) Loseit: $\gamma=1, L=2$; (6) News: $\gamma=1, L=2$.

    \item \textbf{CAGCN*-sc.} (1) Gowalla: $\gamma=1.2$, $l_2$-regularization $1e^{-3}$; (2) Yelp: $\gamma=1.7$, $l_2$-regularization $1e^{-3}$; (3) Amazon: $\gamma=1.7$, $l_2$-regularization $1e^{-3}$; (4) Ml-1M: $\gamma=1$, $l_2$-regularization $1e^{-3}$; (5) Loseit: $\gamma=1, L=2$; (6) News: $\gamma=1, L=2$.

    \item \textbf{CAGCN*-lhn.} (1) Gowalla: $\gamma=1$, $l_2$-regularization $1e^{-3}$; (2) Yelp: $\gamma=1$, $l_2$-regularization $1e^{-3}$; (3) Amazon: $\gamma=1.5$, $l_2$-regularization $1e^{-3}$; (4) Ml-1M: $\gamma=1$, $l_2$-regularization $1e^{-3}$; (5) Loseit: $\gamma=0.5, L=2$; (6) News: $\gamma=1, L=2$.

\end{itemize}

\vspace{-4.5ex}
\subsection{Performance Interpretation}\label{app-performinter}
To demonstrate the generality of our observation in Figure~\ref{fig-imbdeg}, we further perform exactly the same analysis on Yelp (shown in Figure~\ref{fig-imbdegyelp}) and derive almost the same insights: 1) Graph-based recommendation models achieve higher performance than non-graph-based ones for lower degree nodes; 2) the opposite performance trends between NDCG and Recall indicates that different evaluation metrics have different levels of sensitivity to node degrees.

\begin{figure}[htbp!]
     \centering
     \includegraphics[width=0.5\textwidth]{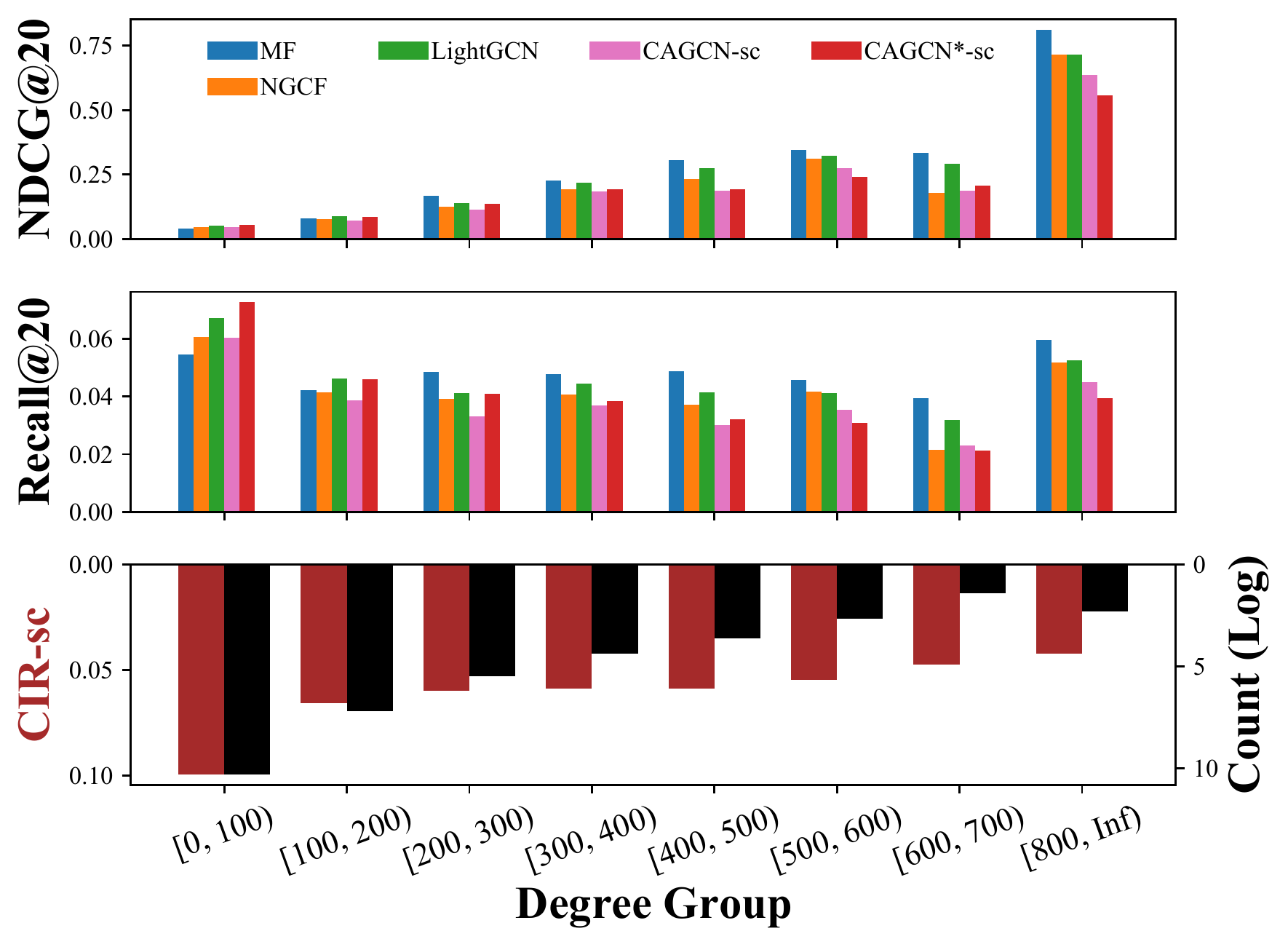}
     \vskip -2.5ex
     \caption{Performance of model w.r.t. node degree on Yelp.}
     \label{fig-imbdegyelp}
     \vspace{-2ex}
\end{figure}

\subsection{Thorough Complexity Analysis}\label{app-complexity2}
Generally compared with the very basic MF, the main computational issue of LightGCN comes from the message-passing which takes $\mathcal{O}(L|\mathcal{E}|F)$ time and $\mathcal{O}(L|\mathcal{V}|F)$ space to save the intermediate node representations. For NGCF, the extra complexity comes from the nonlinear transformation, which takes $\mathcal{O}(L|\mathcal{V}|F^2)$ time and $\mathcal{O}(LF^2)$ space to save the transformation weights. For UltraGCN, the main bottleneck comes from computing the user-user connections, which involves the power of adjacency matrix and hence $\mathcal{O}(|\mathcal{V}|^3)$. Furthermore, as it samples hundreds of negative samples and the optimization is also performed on the user-user connections, then its time complexity would be $\mathcal{O}(r(|\mathcal{E}| + |\mathcal{V}|K)F)$. For CAGCN, since we only consider $2$-hops away connections to compute CIR in Eq.~\eqref{eq-phi-test}(essentially for each center node, we count the number of paths of length 2 from each of its neighbors to its whole neighborhood), the main computational load would be computing the power of adjacency matrix, which takes $\mathcal{O}(|\mathcal{V}|^3)$. Note that for both of our CAGCN and UltraGCN, we can apply Strassens's Algorithm to further reduce the $\mathcal{O}(|\mathcal{V}|^3)$ to $\mathcal{O}(|\mathcal{V}|^{2.8})$ for computing the power of adjacency matrix.

\subsection{Graph Isomorphism}\label{app-graphiso}
We review the concepts of subtree/subgraph-isomorphism~\cite{beyond}.

\begin{definition}
\textbf{Subtree-isomporphism:} $\mathcal{S}_u$ and $\mathcal{S}_i$ are subtree-isomorphic, denoted as $\mathcal{S}_u\cong_{subtree} \mathcal{S}_i$, if there exists a bijective mapping $h: \widetilde{\mathcal{N}}^1_{u} \rightarrow \widetilde{\mathcal{N}}^1_{i}$ such that $h(u) = i$ and $\forall v\in\widetilde{\mathcal{N}}_u^1, h(v)=j, \mathbf{e}_v^l=\mathbf{e}_j^l$.
\end{definition}

\begin{definition}
\textbf{Subgraph-isomporphism:} $\mathcal{S}_u$ and $\mathcal{S}_i$ are subgraph-isomorphic, denoted as $\mathcal{S}_u\cong_{subgraph} \mathcal{S}_i$, if there exists a bijective mapping $h: \widetilde{\mathcal{N}}^1_{u} \rightarrow \widetilde{\mathcal{N}}^1_{i}$ such that $h(u) = i$ and $\forall v_1, v_2\in\widetilde{\mathcal{N}}_u^1, e_{v_1v_2} \in \mathcal{E}_{\mathcal{S}_u}~iff~e_{h(v_1)h(v_2)} \in\mathcal{E}_{\mathcal{S}_i}$ and $\mathbf{e}_{v_1}^l = \mathbf{e}_{h(v_1)}^l, \mathbf{e}_{v_2}^l = \mathbf{e}_{h(v_2)}^l$.
\end{definition}

Corresponding to the backward($\Longleftarrow$) proof of Theorem~\ref{thm-1WL}, here we show two of such graphs $\mathcal{S}_u, \mathcal{S}_u'$, which are subgraph isomorphic but non-bipartite-subgraph-isomorphic. Assuming $u$ and $u'$ have exactly the same neighborhood feature vectors $\mathbf{e}$, then directly propagating according to 1-WL or even considering node degree as the edge weight as GCN~\cite{gcn} can still end up with the same propagated feature for $u$ and $u'$. However, if we leverage JC to calculate CIR as introduced in Appendix~\ref{sec-toponote}, then we would end up with $\{(d_ud_{j_1})^{-0.5}\mathbf{e}, (d_ud_{j_2})^{-0.5}\mathbf{e}, (d_ud_{j_3})^{-0.5}\mathbf{e}\} \ne \{(d_{u'}^{-0.5}d_{j'_1}^{-0.5} + \boldsymbol{\widetilde{\Phi}}_{u'j'_1})\mathbf{e}, (d_{u'}^{-0.5}d_{j'_2}^{-0.5} + \boldsymbol{\widetilde{\Phi}}_{u'j'_2})\mathbf{e}, (d_{u'}^{-0.5}d_{j'_3}^{-0.5} + \boldsymbol{\widetilde{\Phi}}_{u'j'_3})\mathbf{e}\}$. Since $g$ is injective by Lemma~\ref{lemma-injective}, CAGCN would yield two different embeddings for $u$ and $u'$.

\begin{figure}[htbp!]
     \centering
     \includegraphics[width=0.5\textwidth]{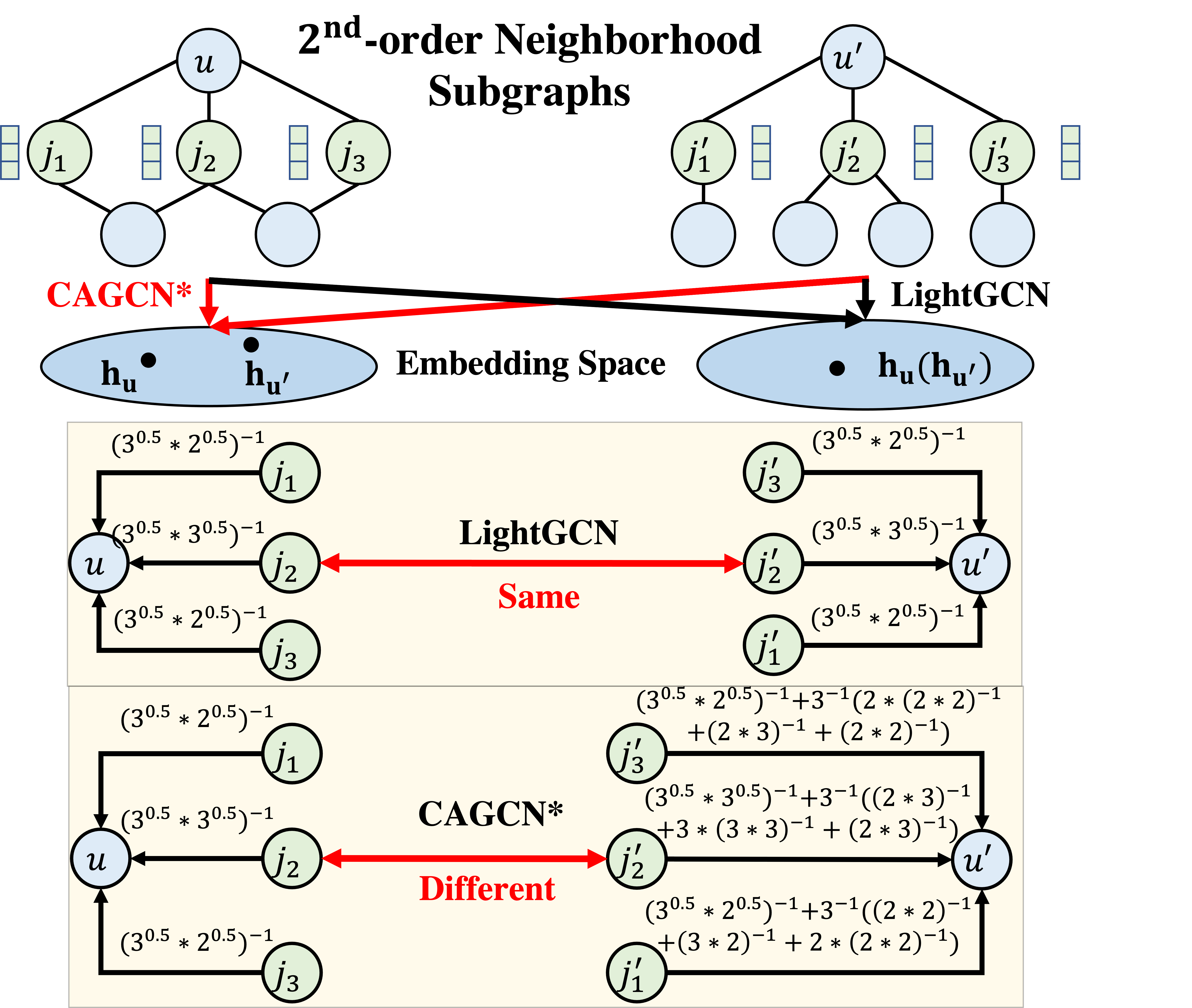}
     \caption{An example showing two neighborhood subgraph $\mathcal{S}_u, \mathcal{S}_{u'}$ that are subgraph-isomorphic but not bipartite-subgraph-isomorphic.}
     \label{fig-example}
     \vspace{-3ex}
\end{figure}

\subsection{Efficiency Comparison}\label{app-effcompare}
Here we use exactly the same setting introduced in Section~\ref{sec-effcompare} and keep track the performance/training time per 5 epochs for Gowalla, Yelp2018, Ml-1M, and Loseit in Figure~\ref{fig-effi_combine}. Clearly, CAGCN* achieves extremely higher performance in significantly less time because the collaboration-aware graph convolution leverages more beneficial collaborations from neighborhoods. Specifically, in Figure~\ref{fig-effi_combine}(c), we observe the slower performance increase of CAGCN* and LightGCN on Ml-1M. We ascribe this to the higher density of Ml-1M as in Table~\ref{tab-dataset} that leads to so much noisy neighboring information. One future direction could be to leverage the CIR to prune the graph of these noisy connections in an iterative fashion as either a preprocessing step or even used throughout training when paired with an attention mechanism (although the latter would come at a significantly longer training time).

\begin{figure*}[htbp!]
     \centering
     \includegraphics[width=0.9\textwidth]{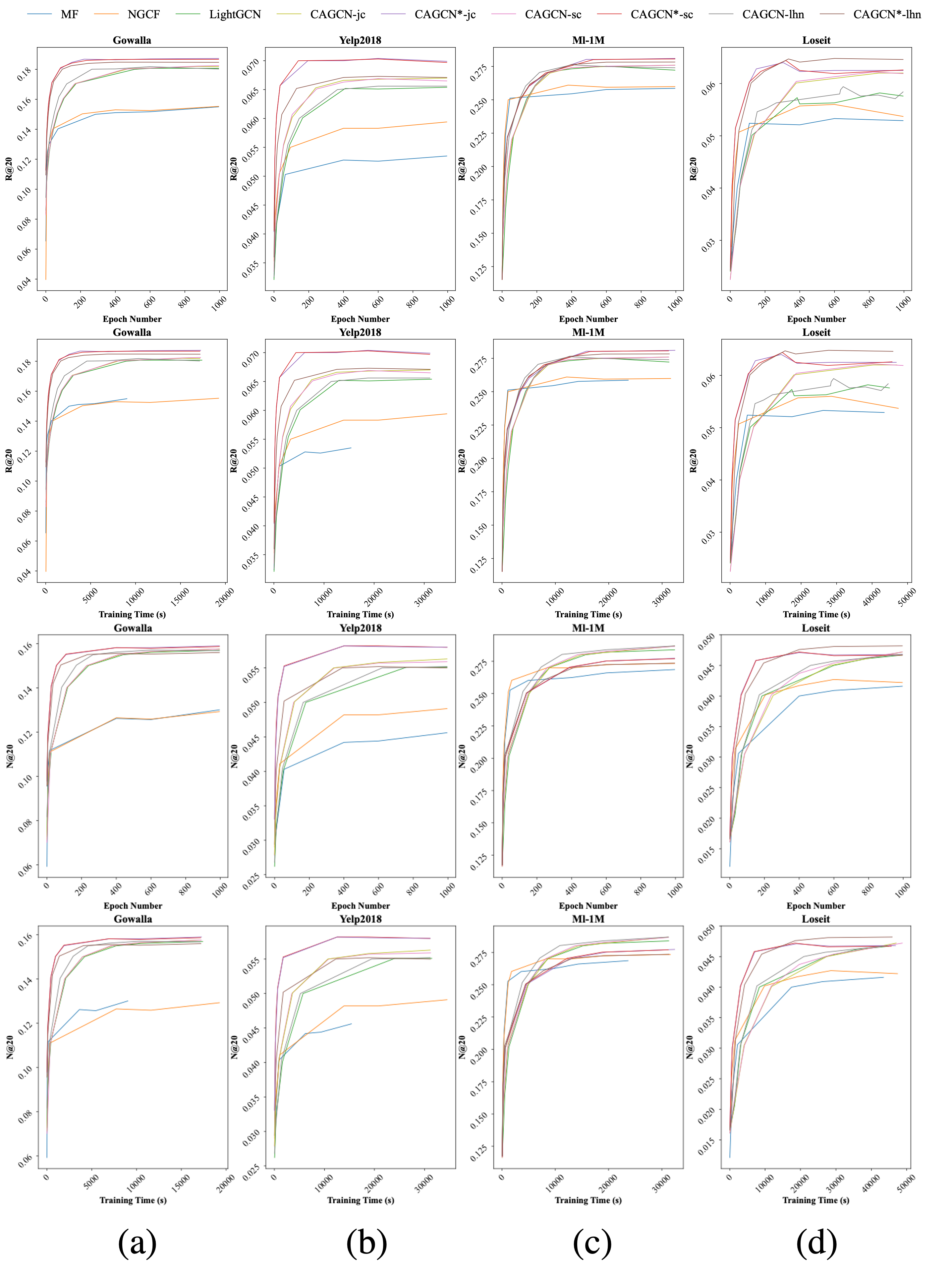}
     \caption{Comparing the training efficiency of each model under R@20 and N@20.}
     \label{fig-effi_combine}
\end{figure*}

\end{document}